\documentclass[journal]{IEEEtran}

\usepackage{graphicx}      

\usepackage{amsthm}
\usepackage{amsmath}
\usepackage{amsbsy}
\usepackage{amssymb}
\usepackage{bm}[nopbm]
\usepackage{multicol}
\usepackage{fancyhdr}
\fancypagestyle{firstpage}{
	\fancyhf{}
	\fancyhead[C]{This work has been submitted to the IEEE for possible publication. Copyright may be transferred without notice, after which this version may no longer be accessible.}
}
\usepackage{color}
\usepackage{caption}
\usepackage{subcaption}

\usepackage{pgfplots}
\pgfplotsset{compat=newest}
\usepackage{tikz}
\usepackage{tikzscale}
\usepackage{placeins}
\usepackage{dblfloatfix}
\usetikzlibrary{external}
\usepackage{array}
\usepackage{booktabs}
\usepackage{multirow}

\newcommand{\vek}[1]{{\mathbf #1}}

\newcommand{\m}[1]{\mathbf{#1}}
\newcommand{\sv}[1]{\boldsymbol{#1}}

\newcommand{\play}[1]{^{(#1)}}

\newcommand{\Tplay}[1]{^{(#1)T}}

\newcommand{\normTwo}[1]{\left\lVert#1\right\rVert_2}

\newtheorem{definition}{Definition}

\newtheorem{lemm}{Lemma}
\theoremstyle{remark}
\newtheorem*{remark}{Remark}
\usepackage{multicol}
\usepackage{fancyhdr}
\fancypagestyle{firstpage}{
	\fancyhf{}
	\fancyhead[C]{This work has been submitted to the IEEE for possible publication. Copyright may be transferred without notice, after which this version may no longer be accessible.}
}

\begin{document}
	\pagenumbering{gobble} 
	\pagestyle{empty} 
	%
	\title{Limited Information Shared Control:\\ A Potential Game Approach}
	
	\author{Balint Varga, Jairo Inga and S\"oren Hohmann
		\thanks{All authors are with the Institute of Control Systems (IRS) at the Karlsruhe Institute of Technology (KIT), 76131 Karlsruhe, Germany
			{\tt\small \{balint.varga, soeren.hohmann\}@kit.edu}}
	}
	
	\maketitle
	\thispagestyle{firstpage}
	\begin{abstract}
	This paper presents a systematic method for the design of a limited information shared control (LISC). LISC is used in applications where not all system states or reference trajectories are measurable by the automation. Typical examples are partially human-controlled systems, in which some subsystems are fully controlled by automation while others are controlled by a human. The proposed systematic design method uses a novel class of games to model human-machine interaction: the near potential differential games (NPDG). We provide a necessary and sufficient condition for the existence of an NPDG and derive an algorithm for finding a NPDG that completely describes a given differential game. The proposed design method is applied to the control of a large vehicle-manipulator system, in which the manipulator is controlled by a human operator and the vehicle is fully automated. The suitability of the NPDG to model differential games is verified in simulations, leading to a faster and more accurate controller design compared to manual tuning. Furthermore, the overall design process is validated in a study with sixteen test subjects, indicating the applicability of the proposed concept in real applications.
	\end{abstract}
	
	\begin{IEEEkeywords}
		Differential Games, Potential Games, Cooperative Shared Control, Limited Information, Vehicle-Manipulator
	\end{IEEEkeywords}
	

	\vspace*{-1mm}
	\section{Introduction}
	\IEEEPARstart{I}n the last years, cooperative shared control applications have shown that they are more efficient than a manually controlled or a fully automated system in many applications~\cite{2018_TopologySharedControl_abbink}. The design of adequate automatic controllers for human-machine shared control systems requires the consideration and understanding of the control action of the human, which increases the acceptance and the performance of the control system \cite{2020_Shared_control_Eraslan}. In literature, there are several problem-specific shared control concepts for applications like working~machines~\cite{2010_BlendedSharedControl_enes}, sea excavator~\cite{2013_HapticSupportBimanual_kuiper}, in the development of cooperative assistive wheelchair \cite{2019_SharedControlSolution_devigne} or automated and autonomous vehicles \cite{2010_NeuromuscularAnalysisGuideline_abbink, 2018_AdaptiveGameTheoreticDecision_tian, 2019_HapticSharedSteering_zwaan, 2020_ReviewSharedControl_marcano}. 
	However, these concepts are problem-specific and consequently they are not easily generalizable.
	
	Therefore, a great effort has been made to develop systematic, model-based design methods for cooperative shared controls e.g.~\cite{2011_OptimalVehicleStability_tamaddoni, 2014_NecessarySufficientConditions_flad, 2015_GameTheoreticModelingSteering_na} or \cite{2019_ModellingOvertakingStrategy_farah}, which ensure the general transferability. They use the theory of differential games to model the human-machine interaction. These model-based approaches are based on the optimality principle of the human control actions \cite{2002_OptimalFeedbackControl_todorov} and the thesis that the haptic interactions can be modelled as a differential game in its Nash Equilibrium (NE), see~\cite{2009_NashEquilibriaMultiAgent_braun}, which has a widespread experimental evidence. Therefore, differential games are useful tools for the design of a shared control.
		
	These methods in the literature presume full information about the controlled system, therefore they are addressed as \textit{full information shared control} (FISC) in this paper. Such full information is however usually not given: First, some system states or reference trajectories are controlled by the human and may not be measurable for the automation. This can be caused by the sparing of sensors due to cost reduction or the technical impossibility of a measurement (e.g. the lack of localization in deep ocean applications~\cite{2014_AUVNavigationLocalization_paull} or working machines in unstructured working environment~\cite{2011_OpenloopControlExperiments_moralesa}). Second, the goals of the human are neither available nor exactly predictable in real world applications. To solve the challenge of the non-measurable references, a \textit{limited information shared control} (LISC) is proposed in \cite{2019_ControlLargeVehicleManipulators_varga} and \cite{2020_LimitedInformationCooperativeShared_varga}. However, in these earlier works present only heuristic parameter design of the LISC, but no systematic approach. 
	
	Therefore, in this paper, this research gap is addressed by introducing a systematic design of the LISC, which makes LISC more generally applicable. The idea is the usage of a FISC as a baseline for the design procedure of the LISC. This approach is justified by the following reasons: The FISC can be applied in a test-area or in a simulation environment, where all the system states are measurable to identify and design the LISC with a human operator. Then the LISC is suitable for real world applications, where full information of the system is not given. These design steps are illustrated in Fig. \ref{fig:design_steps}.
	
	For this systematic design, cooperative setups are modelled for the first time as \textit{potential games} (PGs). A PG provides a more compact and simpler representation, in which the original game is substituted by a single-player optimal control problem, see~\cite{1996_PotentialGames_monderer}. The optimal input of the single player yields the NE of the game. In our work, the so-called near PGs, introduced for static games in~\cite{2013_NearPotentialGamesGeometry_candogan}, are extended to differential games for the first time. To the best of the authors' knowledge, no research work deals with the theory or the application of the \textit{near potential differential games} (NPDGs). The contributions of this paper are as follows: 
	1) Introduction of the NPDG as an extension of the static near PGs, 2) Modelling human-machine interaction based on NPDG providing a simpler model than an N-player model, 3) Systematic design of a LISC replacing the manual tuning and 4) An experimental case study involving human-machine shared control of a large vehicle-manipulator (VM) in order to compare FISC to LISC.  
	
	In the following, Section~\ref{sec:state-of-art} presents the preliminaries of modelling human-machine interactions by means of game theory and the concept of the LISC. 
	\begin{figure}[t!]
		\centering
		
		\includegraphics[width=0.9\linewidth]{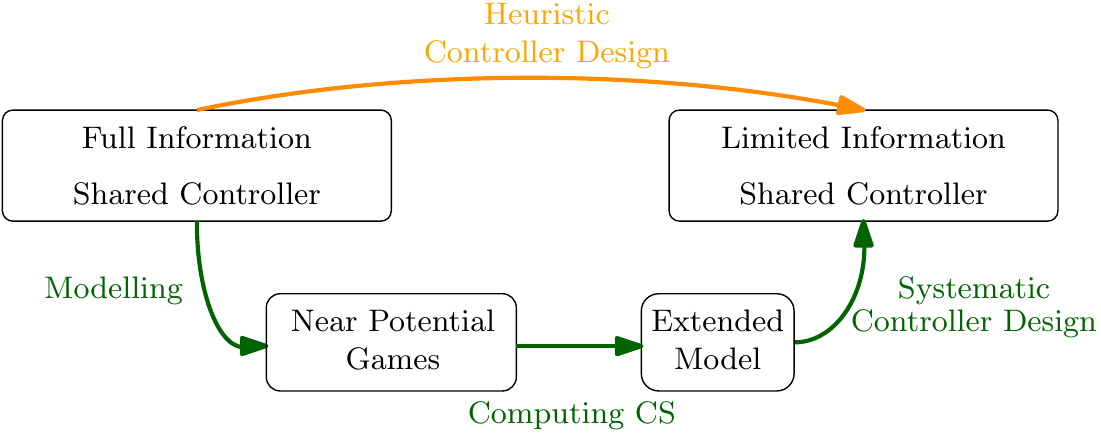}
		
		\caption{Illustration of the design steps.}
		\label{fig:design_steps}
	\end{figure}
	In Section~\ref{sec:Near_pot_games}, the concept of the NPDG is presented. A systematic identification of the CS through NPDG and a systematic design of LISC are given in Section~\ref{sec:sys_derivation}. In Section~\ref{sec:verif_sim}, the concept of the NPDGs is verified in simulations of a large VM. In Section~\ref{sec:validation}, the LISC is validated on a demonstrator with sixteen test subjects and a comparison between FISC and LISC is provided. Finally, Section~\ref{sec:conclu} concludes the contribution. Due to the considerable number of abbreviations, Table \ref{table:nomenclatures} provides a summary of them.
	
	\section{Preliminaries} \label{sec:state-of-art}
	In this section, the state-of-the-art of shared control applications using differential games and the core idea of the LISC are presented. Furthermore, the state-of-the-art of the exact PDGs and the solution steps of the proposed design are given.
	\begin{table}[b!] 
		\small
		\centering
		\caption{Nomenclature of the Abbreviations}
		\begin{tabular}{l|c}
			\multicolumn{1}{c|}{Terminus} & \multicolumn{1}{c}{Acronym} \\
			\hline
			\hline
			Nash Equilibrium & NE \\
			Limited Information shared control & LISC \\	
			Full Information shared control & FISC \\
			Potential Game & PG \\
			Potential Differential Game & PDG \\
			Near Potential Differential Game & NPDG \\
			Linear-Quadratic & LQ \\
			Cooperation State & CS \\
			Vehicle Manipulator & VM \\
			Two One-Sided t-Test & TOST \\
			\hline
		\end{tabular} 
		\label{table:nomenclatures}
	\end{table}
	\subsection{Modelling Shared Controls by Differential Games}
	In shared control applications, the actions of the automation and the human are linked: They combine their efforts reaching a common goal \cite{2018_TopologySharedControl_abbink}. 
	A survey of the design of shared controls is given e.g.~in~\cite{2017_SharedCooperativeControl_flemisch}. 
	
	For a systematic design of a shared control, an important aspect is the modelling of the human. A hypothesis to model human motion is based on the theory of optimal control: The human movements are the results of the dynamic optimization of an objective function $J\play{\mathrm{h}}$, see e.g~\cite{2002_OptimalFeedbackControl_todorov}. 
	A quadratic cost function with an underlying linear system dynamic is common in control theory for modelling human actions (see~\cite{2015_SolutionsInverseLQR_priess, 2019_InverseDiscountedbasedLQR_el-hussieny})
	\begin{align} \label{eq:hum_quad_cost}
		&J\play{\mathrm{h}} = \frac{1}{2} \int_{0}^{T_\mathrm{end}} \sv{x}^T \vek{Q}\play{\mathrm{h}} \sv{x} + {\sv{u}\play{\mathrm{h}}}^T \vek{R}\play{\mathrm{h}} \sv{u}\play{\mathrm{h}}\text{ d} t, \\
		& \;\text{s. t. } \; \; \; 	\dot{\sv{x}}(t) = \vek{A}\sv{x}(t) + \sum_{i \in \mathcal{P} } \vek{B}\play{i}\sv{u}\play{i}(t), \label{eq:lin_system}
	\end{align}
	where {${\vek{Q}\play{\mathrm{h}} \geq 0} \in \mathbb{R}^{n \times n}$} and ${\vek{R}\play{\mathrm{h}} > 0 \in \mathbb{R}^{p_\mathrm{h} \times p_\mathrm{h}}}$, with $J\play{\mathrm{h}}: \mathbb{R}^n \times \mathbb{R}^{p_\mathrm{h}}  \rightarrow \mathbb{R}$. The dynamic system in is $\sv{f}:\mathbb{R}^+ \times \mathbb{R}^n \times \mathbb{R}^{p_\mathrm{h}} \rightarrow \mathbb{R}^n$ and where $\sv{x} \in \mathbb{R}^n$ and $\sv{u}\play{\mathrm{h}} \in \mathbb{R}^{p_\mathrm{h}}$ are the system states and the inputs of the human, respectively.
	In~\cite{2009_NashEquilibriaMultiAgent_braun}, it is shown that haptic human interactions converges to a NE, therefore differential games are useful tools for the design of a shared control.
	With (\ref{eq:hum_quad_cost}) and (\ref{eq:lin_system}), the systematic design of a shared control is possible by means of game theory, see e.g. \cite{2014_NecessarySufficientConditions_flad, 2015_GameTheoreticModelingSteering_na}. In~\cite{2014_NecessarySufficientConditions_flad}, a global objective function $J\play{g}$ assumed to be given. This global objective is specified by the system designers to adapt the behaviour of the overall control loop to higher-level requirements. In that work, a quadratic cost function of the automation is suggested 
	\begin{equation} \label{eq:cost_fucntion}
		J\play{a}\!=\!\frac{1}{2}\int_{0}^{T_\mathrm{end}} \!\!\! \sv{x}(t)^T  \vek{Q}\play{a}\sv{x}(t)  \!+\!   \sum_{j \in \mathcal{P}} {\sv{u}\play{j}(t)}^T \vek{R}\play{aj}\sv{u}\play{j}(t)  \text{ d}t,
	\end{equation}
	with $J\play{a}: \mathbb{R}^n \times \mathbb{R}^{p_\mathrm{a}} \times \mathbb{R}^{p_\mathrm{h}} \rightarrow \mathbb{R}$, where $j \in \mathcal{P}=\{\mathrm{a},\;\mathrm{h}\}$, ${{\vek{Q}\play{a} \geq 0} \in \mathbb{R}^{n \times n}}$ and ${\vek{R}\play{\mathrm{a}j} > 0 \in \mathbb{R}^{p_\mathrm{a} \times p_\mathrm{j}}}$. 
	With (\ref{eq:hum_quad_cost}), (\ref{eq:lin_system}) and (\ref{eq:cost_fucntion}), a differential game $\Gamma_\text{d}$ is specified, with the corresponding player hamiltonian functions
	\begin{align} \label{eq:ham_SOA}
			H\play{i} =& \frac{1}{2}\sv{x}(t)^T  \vek{Q}\play{a}\sv{x}(t) \\ \nonumber
			&+\frac{1}{2}   \sum\limits_{j \in \mathcal{P}} {\sv{u}\play{j}(t)}^T \vek{R}\play{ij}\sv{u}\play{j}(t)	+\sv{\lambda}\Tplay{i}\sv{f}(t).
	\end{align}
	From (\ref{eq:ham_SOA}), the NE with a linear feedback control law of the players 
	$ 
	\sv{u}\play{i} = -\vek{K}\play{i}\sv{x}, 
	$ 
	can be computed. The necessary and sufficient condition for a NE is the existence of a solution $\vek{P}^{(i)}$ of the coupled algebraic Riccati equation \cite[p.~295]{2005_LQDynamicOptimization_engwerda}
	\begin{align} \label{eq:sol_i_Ric} \nonumber
		\vek{0} &= \vek{A}_c^T \vek{P}\play{i} +  \vek{P}\play{i} \vek{A}_c + \vek{Q}\play{i} \\ 
		&- \sum_{j \in \mathcal{P}} \vek{P}\play{i}\vek{B}\play{i}{\vek{R}\play{ii}}^{-1}\vek{R}\play{ij}{\vek{R}\play{ii}}^{-1} {\vek{B}\play{i}}^{T} \vek{P}\play{i}
	\end{align}
	where ${\vek{A}_c = \vek{A} - \sum_{i \in \mathcal{P} } \vek{B}\play{i}\vek{K}\play{i}}$ is the matrix of the closed loop system dynamics. The solution of~(\ref{eq:sol_i_Ric}) provides the feedback control gains ${	\vek{K}\play{i} = {\vek{R}\play{ii}}^{-1}\vek{P}\play{i}\vek{B}\play{i}.}$
	
	\subsection{Concept of the Limited Information Shared Control} \label{sec:concept_cooperation_sattes}
	In this subsection, a brief overview of the LISC-concept with linear system model is given, which is introduced by the authors in \cite{2020_LimitedInformationCooperativeShared_varga}. It is assumed that the system is modelled in the so-called Fr\'enet frame, which is a common approach for mobile robots \cite[Ch.~49.2]{2016_SpringerHandbookRobotics_brunosiciliano}. This means that the system states are given relative to the reference path and therefore the goal state is always $\sv{x}_{\mathrm{goal}} = \sv{0}$.  
	In the following, we consider a linear system dynamics in Fr\'enet frame, which can split into automation-controlled (measurable) $\sv{x}_m \in \mathbb{R}^{n-k}$ and human-controlled (non-measurable) states $\sv{x}_{nm} \in \mathbb{R}^k$. It is assumed that the human-controlled system state has no impact on the automation-controlled state states. Thus \mbox{(\ref{eq:lin_system})} is transformed with these separated states to the system model
	\begin{align} \label{eq:linearStateSpaceCoupMotion} \nonumber
		\begin{bmatrix}
			\dot{\sv{x}}_m(t) \\
			\dot{\sv{x}}_{nm}(t)
		\end{bmatrix}
		= 
		&\begin{bmatrix}
			\mathbf{A}^m_{m} & \mathbf{0} \\
			\mathbf{A}^m_{nm} & \mathbf{A}^{nm}_{nm} 
		\end{bmatrix}
		\begin{bmatrix}
			\sv{x}_m(t) \\
			\sv{x}_{nm}(t)
		\end{bmatrix} \\
		+ 
		& \,
		\mathbf{B}\play{\mathrm{a}} 
		\sv{u}\play{\mathrm{a}}(t) 
		+ 
		\mathbf{B}\play{\mathrm{h}} 
		\sv{u}\play{\mathrm{h}}(t). 
	\end{align}
	The automation controls only the measurable system part. The human operator controls the non-measurable system part, and therefore these human-controlled states are not measured.
	On the other hand, to support the operator, the non-measurable system part, which is influenced by the human, needs to be considered in controller design. 
	The proposed LISC solves this problem by modifying the control model for the design procedure. The idea is the introduction of an artificial state, the so-called cooperation state (CS), for reconstructing $\sv{x}_{nm}$ using the inputs of the human $\sv{u}\play{\mathrm{h}}(t)$ and the automation $\sv{u}\play{\mathrm{a}}(t)$. The human's input is function of the non-measurable system part. The CS provides an inversion of this function with the additional consideration that $\sv{u}\play{a}$ has an impact on $\sv{x}_{nm}$.
	\begin{definition}[Cooperation state \cite{2020_ValidationCooperativeSharedControl_vargab}] \label{def:coop}
		Consider a system with known system dynamics and with a measurable and with a non-measurable system parts~(\ref{eq:linearStateSpaceCoupMotion}). The cooperation state is defined as 
		\begin{equation} \label{eq:linCooperation_state}
			\sv{x}_\kappa(t) = {\Xi}\play{\mathrm{a}} \sv{u}\play{\mathrm{a}}(t) + {\Xi}\play{\mathrm{h}} \sv{u}\play{\mathrm{h}}(t)  \, \in \mathbb{R}^{k},
		\end{equation}
		where the matrices $\bm{\Xi}\play{\mathrm{a}} \in \mathbb{R}^{k \times p_a}$ and $\bm{\Xi}\play{\mathrm{h}} \in \mathbb{R}^{k \times p_h}$ are design parameters. 
	\end{definition}
	With~(\ref{eq:linCooperation_state}), an extended state vector is introduced ${\sv{x}_e(t) = [\sv{x}_{m}   \;
		\sv{u}\play{\mathrm{a}} \;
		\sv{x}_{\kappa} ]^T}, \, \in \mathbb{R}^{n+p_h}$ leading to an extended system dynamics
	\begin{equation} \label{eq:extended_system_eqaution}
		\!\resizebox{.92\hsize}{!}{$
			\begin{bmatrix}  
				\dot{\sv{x}}_{m}  \\
				\dot{\sv{u}}\play{\mathrm{a}} \\
				\dot{\sv{x}}_{\kappa} 
			\end{bmatrix}
			\! = \!
			\setlength{\arraycolsep}{1pt}
			\renewcommand{\arraystretch}{0.99}
			\begin{bmatrix}
				\mathbf{A}_{m} &\mathbf{B}\play{\mathrm{a}} & \sv{0} \\
				\sv{0} & \sv{0}& \sv{0}\\
				\sv{0} & \sv{0}& \sv{0}
			\end{bmatrix} \!
			\begin{bmatrix}  
				\sv{x}_{m}  \\
				\sv{u}\play{\mathrm{a}} \\
				\sv{x}_{\kappa} 
			\end{bmatrix}
			\! + \!
			\begin{bmatrix}  
				\sv{0} \\
				\sv{1}  \\
				\mathbf{\Xi}\play{\mathrm{a}}
			\end{bmatrix} \!
			\dot{\sv{u}}\play{\mathrm{a}} 
			\! + \!
			\begin{bmatrix}  
				\sv{0} \\
				\sv{0} \\
				\mathbf{\Xi}\play{\mathrm{h}}
			\end{bmatrix} \!
			\dot{\sv{u}}\play{\mathrm{h}}\!,$} \! \! \!
	\end{equation}
	where the derivate of the original system input $\dot{\sv{u}}\play{a}$ is taken into account only for the design procedure of the LISC. The structure of (\ref{eq:extended_system_eqaution}) shows that $\sv{x}_\kappa$ can effect $\sv{x}_m$ and $\sv{u}\play{a}$, if feedback controller is designed. 
	This modified model (\ref{eq:extended_system_eqaution}) and the cost function
	\begin{equation} \label{eq:extended_cost_function}
		J\play{\mathrm{a}}_{\text{LI}} \!= \!\int_{0}^{T_\mathrm{end}} \! \! {\sv{x}_e(t)}^T \vek{Q}_{\text{LI}}\play{\mathrm{a}} \sv{x}_e(t) + \dot{\sv{u}}\Tplay{a}\!(t) \vek{R}_{\text{LI}}\play{\mathrm{a}} \dot{\sv{u}}\play{\mathrm{a}}(t) \, \text{d} t,
	\end{equation}
	is used to formulate an LQR control problem, which excludes the non-measurable states~$\sv{x}_{nm}$ and enables the systematic design of a cooperative shared control. By solving (\ref{eq:extended_cost_function}), subject to (\ref{eq:extended_system_eqaution}), a linear control law is obtained
	\begin{equation} \label{eq:newMethodLinController}
		\dot{\sv{u}}\play{\mathrm{a}}_\mathrm{LI}(t) = - \vek{K}\play{\mathrm{a}}_{\text{LI}} \cdot \sv{x}_e(t),
	\end{equation}
	from which the the original system input is computed by
	\begin{equation}  \label{eq:newMethodLinControllerIntegral}
		\sv{u}\play{\mathrm{a}}_\mathrm{LI}(t) = \int_{0}^{t} \dot{\sv{u}}\play{\mathrm{a}}_\mathrm{LI}(\tau)  \; \text{d} \tau.
	\end{equation} 
	In (\ref{eq:extended_system_eqaution}) the initial value of original system input $\sv{u}\play{a}$ is assumed $\sv{u}\play{a}(0) = \sv{0}$ meaning that the original system input signal is zero at the beginning. This is a plausible assumption, since the controller can be initialized to zero without loss of generality.
	
	\subsection{Exact LQ Potential Differential Games}
	The following preliminaries are based on \cite{2016_SurveyStaticDynamic_gonzalez-sanchez}. PGs were first introduced in \cite{1996_PotentialGames_monderer}. The benefit of the PGs that they include all information about the original game, despite the simple description of the game.  
	\begin{definition}[Exact LQ PDG]
		Consider an LQ differential game $\Gamma_d$ with system dynamics (\ref{eq:lin_system}), 
		quadratic cost functions~(\ref{eq:cost_fucntion}) and Hamiltonian functions (\ref{eq:ham_SOA}). Consider further an LQ optimal control problem with (\ref{eq:lin_system}) and the cost function
		\begin{equation} \label{eq:pot_cost_function}
			J\play{p}=\frac{1}{2}\int_{0}^{T_\mathrm{end}} \sv{x}(t)^T \vek{Q}\play{p}\sv{x}(t) +  \sv{u}(t)^T\vek{R}\play{p}\sv{u}(t) \text{ d}t,
		\end{equation} 
		as well as the corresponding Hamilton function 
		\begin{equation} \label{eq:potential_hamil_quad_def}
			H\play{p} = \frac{1}{2}{\sv{x}(t)}^T \vek{Q}\play{p}\sv{x}(t) +  \frac{1}{2}{\sv{u}(t)}^T\vek{R}\play{p}\sv{u}(t) + \sv{\lambda}\Tplay{p}\sv{f}(t),
		\end{equation}
		with $\sv{u} = \left[\sv{u}\play{a}, \, \sv{u}\play{h}\right]$, where the matrices $\vek{Q}\play{p}$ and $\vek{R}\play{p}$ are positive semi-definite and positive definite, respectively.
		If the condition 
		\begin{equation} \label{eq:def_e_pot}
			\frac{\partial H\play{p}}{\partial \sv{u}\play{i}(t)} = \frac{\partial H\play{i}}{\partial \sv{u}\play{i}(t)}
		\end{equation} 
		holds for $i \in \mathcal{P}$, then the LQ differential game $\Gamma$ is an exact PDG with respect to $J\play{p}$. The index $(p)$ symbolizes the PG. 
	\end{definition}
	In \cite{2018_PotentialDifferentialGames_fonseca-morales}, a general overview of the exact PDGs is given. The main drawback of the exact PDGs is that their use is restricted to some special classes of problems, where the system or input matrices must are special. The examples discussed in the literature have always some kind of special structure e.g. weak or no relation between the system states: Each input signal can influence only one system state or in other examples, the cost functions of the players contain of sparse matrices, such that (\ref{eq:def_e_pot}) is easily fulfilled, see \cite{2018_PotentialDifferentialGames_fonseca-morales, 2016_PotentialGameTheory_la}. 
	
	\subsection{Design Problem and the Solution Steps}
	The main challenge is how to omit the manual tuning of the LISC, which is time consuming. Therefore, a systematic identification of the parameters of the matrices $\bm{\Xi}\play{\mathrm{a}}$ and $\bm{\Xi}\play{\mathrm{h}}$ in (\ref{eq:linCooperation_state}) is crucial and necessary, because a manual tuning hampers the transferability to other systems.
	Furthermore, a systematic approach is necessary to find a suitable $J\play{\mathrm{a}}_{\mathrm{LI}}$ that provides the desired behaviour and omits manual tuning, cf. Section \ref{sec:sys_derivation}-B.	
	To solve these challenges, this paper presents a design process with the following steps (cf. Fig. \ref{fig:design_steps}) and presented in detail in Section \ref{sec:sys_derivation}.
	\begin{itemize}
		\item[1] Design a FISC according to \cite{2014_NecessarySufficientConditions_flad}, with the desired behaviour of the overall system, which is set with the global objective goal function $J\play{g}$, see first part of \ref{sec:sys_derivation}-B,
		\item[2] Identifying an NPDG, which models the cooperative setup designed in the first step, see \ref{sec:Near_pot_games}-C,
		\item[3] Design of the parameter matrices $\bm{\Xi}\play{\mathrm{a}}$ and $\bm{\Xi}\play{\mathrm{h}}$ of the CS, see \ref{sec:sys_derivation}-A,
		\item[4] Design of $J\play{\mathrm{a}}_{\text{LI}}$ leading to the desired behaviour of the overall system defined in the first step, see second part in \ref{sec:sys_derivation}-B.
	\end{itemize}
	
	\section{Near Potential Differential Games} \label{sec:Near_pot_games}
	This section presents the novel concept of the NPDGs and the analysis of the NPDGs' dynamics. As mentioned earlier, NPDG identification is a key part of the design process and NPDGs are transferred from the static to the dynamic case for the first time in this paper. A special subclass of the PGs is the near PG introduced in~\cite{2013_NearPotentialGamesGeometry_candogan}, where it has been shown that a near PG has similar dynamics as a PG, which means that the input and states trajectories are similar. There, only static games are taken into account. In the following, we give an extension to differential games.
	
	\subsection{Distance to LQ NPDG}
	Similar to static case in \cite{2013_NearPotentialGamesGeometry_candogan}, we introduce a distance notation for NPDG.
	\begin{definition}[Differential Distance]
		Let a LQ PDG $\Gamma\play{p}_d$ with the potential function~(\ref{eq:pot_cost_function}) be defined. Let us further define an arbitrary LQ differential game $\Gamma_{d}$ with the cost functions of the players defined as in~(\ref{eq:cost_fucntion}). The \textit{Differential Distance} (DD) between $\Gamma\play{p}_d$ and $\Gamma_{d}$ is defined as
		\begin{equation} \label{eq:diff_distance}
			\sigma\play{i}_d(t):= \left\lVert\frac{\partial H\play{p}(t)}{\partial \sv{u}\play{i}(t)} - \frac{\partial H\play{i}(t)}{\partial \sv{u}\play{i}(t)}\right\rVert_2, \; i \in \mathcal{P}.
		\end{equation} 
	\end{definition}
	Further, we assume that the individual control laws are obtained from the solution of the coupled Riccati equations~(\ref{eq:sol_i_Ric}) over an infinite time horizon, which leads to the closed-loop system dynamics
	\begin{align}\label{eq:closed_loop_linear_system}
		\dot{\sv{x}}(t) = \vek{A}^*_c\sv{x}(t), \; \; \sv{x}(t_0) = \sv{x}_0,
	\end{align}
	where ${\vek{A}^*_c = \vek{A} - \sum_{i \in \mathcal{P}} \vek{B}\play{i}{\vek{R}\play{i}}^{-1} {\vek{B}\play{i}}^T \vek{P}\play{i}}$ and with the unique solution 
	\begin{equation} \label{eq:sol_NE}
		\sv{x}^*(t) = e^{\vek{A}^*_c \cdot t}\sv{x}_0.
	\end{equation}
	Analogously, let us use the control law $\vek{K} = {\vek{R}\play{p}}^{-1} {\vek{B}}^T \vek{P}\play{p}$, gained from the optimization of the potential function (\ref{eq:pot_cost_function}) to take into consideration the closed loop system dynamics
	\begin{align}\label{eq:closed_loop_linear_system_pot}
		\dot{\sv{x}}\play{p}(t) = \vek{A}\play{p}_c\sv{x}(t), \; \; \sv{x}\play{p}(t_0) = \sv{x}\play{p}_0,
	\end{align}
	where $\vek{A}\play{p}_c$ is computed such that ${\vek{A}\play{p}_c = \vek{A} - \vek{B}{\vek{R}\play{p}}^{-1} {\vek{B}}^T \vek{P}\play{p}}$ with ${\vek{B}=[\vek{B}\play{\mathrm{a}}, \vek{B}\play{\mathrm{h}}]}$. The solution of (\ref{eq:closed_loop_linear_system_pot}) is
	\begin{equation} \label{eq:sol_pot}
		\sv{x}\play{p}(t) = e^{\vek{A}\play{p}_c \cdot t}\sv{x}\play{p}_0.
	\end{equation}
	\begin{definition}[Near Potential Differential Game]
		A differential game is an NPDG if the DD is 
		\begin{align}
			\underset{i}{\mathrm{max }}\normTwo{\sigma\play{i}_d(t)}  < \Delta, \; \; i \in \mathcal{P},
		\end{align}
		where $\Delta>0$ is some small arbitrary constant. 
	\end{definition}
	\begin{lemm}[LQ-NPDG]
		Let an LQ differential game with NE state trajectories (\ref{eq:sol_NE}) be given. Furthermore, let a PDG with NE state trajectories (\ref{eq:sol_pot}) be given. The game $\Gamma_d$ is a LQ-NPDG, if
		\begin{align}
			\underset{i}{\mathrm{max }}\left\lVert{\vek{B}\play{i}}^T\vek{P}\play{p} - {\vek{B}\play{i}}^T\vek{P}\play{i} \right\rVert_2 \cdot \sv{x}_{\mathrm{max}} < \Delta
		\end{align}
		holds, where ${\sv{x}_\mathrm{max} = \mathrm{max }\left(\left\lVert\sv{x}^*(t)\right\rVert_2, \left\lVert \sv{x}\play{p}(t)\right\rVert_2\right)}$ is the maximum magnitude of the state vectors. 
	\end{lemm}
	\begin{proof}
		For the proof, the two partial derivatives of the Hamiltonians in (\ref{eq:diff_distance}) are analysed using the definition of the LQ NPDGs.	
		The derivatives of $H\play{i}$ is expressed as
		\begin{equation} \label{eq:deriv_player_i}
			\frac{\partial H\play{i}(t)}{\partial \sv{u}\play{i}(t)} = \m{R}\play{i}\sv{u}\play{i}(t) + \m{B}\Tplay{i}\sv{\lambda}\play{i}(t),
		\end{equation}
		which holds for $i \in \mathcal{P}$, and $\sv{\lambda}\play{i}(t) = \m{P}\play{i}\sv{x}^*(t)$ can be substituted. 
		For the further analysis, a suboptimal control law is assumed and from this the optimal control law is obtained. The applied sub-optimal control laws of the original game is
		\begin{equation} \label{eq:sub-opt-control-law}
			\sv{u}\play{i}(t) = -(1 + \varepsilon_c){\m{R}\play{i}}^{-1}\m{B}\Tplay{i}\m{P}\play{i}\sv{x}^*(t), 
		\end{equation}
		where  $\varepsilon_c \neq 0$ is arbitrary constant and $i \in \mathcal{P}$. The optimal control law is obtained with $\varepsilon_c \rightarrow 0$. The control law (\ref{eq:sub-opt-control-law}) yields the behaviour of players around the optimal solution.
		Substituting (\ref{eq:sub-opt-control-law}) in (\ref{eq:deriv_player_i}) gives
		\begin{equation*}
			\! \! \! \!\frac{\partial H\play{i}(t)}{\partial \sv{u}\play{i}(t)}\, \!\resizebox{.865\hsize}{!}{$=\! -\m{R}\play{i}(1 + \varepsilon_c){\m{R}\play{i}}^{-1}\m{B}\Tplay{i}\m{P}\play{i}\sv{x}^*(t) +\m{B}\Tplay{i}\m{P}\play{i}\sv{x}^*(t),$} \! \!
		\end{equation*}
		which can be simplified to 
		\begin{equation} \label{eq:pr_2_res_2}
			\frac{\partial H\play{i}(t)}{\partial \sv{u}\play{i}(t)} = -\varepsilon_c\vek{B}\Tplay{i}\vek{P}\play{i}\sv{x}^*(t).
		\end{equation}
		Analogously, the derivatives of the Hamiltonian of the potential function is obtained
		\begin{equation} \label{eq:diff_Hp}
			\frac{\partial H\play{p}(t)}{\partial \sv{u}\play{i}(t)} = -\varepsilon_c{\vek{B}\play{i}}^T\vek{P}\play{p}\sv{x}\play{p}(t).
		\end{equation}
		The simplification of the derivatives of the Hamiltonians and a substitution in the definition of DD (\ref{eq:diff_distance}) leads to the proof of the lemma.
	\end{proof}
	If the DD $\sv{\sigma}_d$ between the original differential game and the substituting PDG is \textit{small}, similar to the static case in \cite{2013_NearPotentialGamesGeometry_candogan}, the resulting closed loop behaviours are also similar. 
	\subsection{Dynamics of the LQ NPDGs} 
	This subsection provides the connection between the upper estimation ($\varepsilon$) of the errors between the trajectories $\sv{x}^*(t)$ and $\sv{x}\play{p}(t)$ and the DD of the two games ($\Delta$).
	Here for, we introduce a notation for the difference of the closed-loop behaviours:
	\begin{definition}[Closed-Loop System Matrix Error] \label{def:SDD}
		Let an LQ differential game $\Gamma_d$ with the control law of the NE and an LQ NPDG  $\Gamma\play{p}_d$ be given. The Closed-loop System Matrix Error (CSME) between $\Gamma\play{p}$ and $\Gamma\play{p}_d$ is denoted as
		\begin{align}
			\Delta \vek{E} &:= \normTwo{\vek{A}^*_c - \vek{A}\play{p}_c}
		\end{align}
	\end{definition} 
	\begin{lemm}[Boundedness of NPDGs] \label{lem:npdg_boundedness}
		Given an LQ differential game $\Gamma_{d}$ and an LQ PDG $\Gamma\play{p}_{d}$ defined by (\ref{eq:cost_fucntion}) and (\ref{eq:pot_cost_function}), respectively, where the underlying dynamics system is given in (\ref{eq:lin_system}). It is assumed that the initial values are identical:
		\begin{align} \label{eq:proof_init_value}
			{\sv{x}\play{p}(t_0) = \sv{x}^*(t_0) = \sv{x}_0}.
		\end{align}
		If the DD between $\Gamma_{d}$ and $\Gamma\play{p}_{d}$ is bounded in the time interval $t \in [t_0,t_1],$ such that 
		\begin{equation} \label{eq:estim_delta}
			\underset{i}{\mathrm{max }}\normTwo{\sigma\play{i}_d(t)}  < \sigma_\mathrm{max} = \Delta, \; \; i \in \mathcal{P},
		\end{equation}
		then the deviation between their trajectories are also bounded, with 
		\begin{align} \label{eq:dif_traj_with_eps}
			\normTwo{\sv{x}\play{p}(t)-\sv{x}^*(t)} \leq \varepsilon(\Delta), \; \; \forall t \in [t_0,t_1]. 
		\end{align}
	\end{lemm}
	\begin{proof} 
		From the solution of the differential equations (\ref{eq:closed_loop_linear_system}) and (\ref{eq:closed_loop_linear_system_pot}), 
		$
		\normTwo{\sv{x}\play{p}(t)-\sv{x}(t)}= \normTwo{e^{\vek{A}^*_ct}\sv{x}_0 - e^{\vek{A}^*_ct}\sv{x}_0}.
		$ is obtained.
		As (\ref{eq:proof_init_value}) holds, 
		$$
		\normTwo{\sv{x}\play{p}(t)-\sv{x}^*(t)}= \normTwo{e^{\vek{A}\play{p}_c \cdot t}\sv{x}_0 \left(e^{\Delta \vek{E} \cdot t} - 1\right)}
		$$ 
		is obtained. In the following an upper estimation of $\Delta \vek{E}$ is sought. Substituting (\ref{eq:sol_NE}) and (\ref{eq:closed_loop_linear_system_pot}) in the Definition \ref{def:SDD}, 
		\begin{align*}
			\Delta \vek{E} = \vek{B}\left({\vek{R}\play{p}}^{-1} {\vek{B}}^T \vek{P}\play{p} - \sum_{i \in \mathcal{P}} {\vek{R}\play{i}}^{-1} {\vek{B}\play{i}}^T \vek{P}\play{i}\right) 
		\end{align*} is obtained. There is a manifold of quadratic cost functions of the players~(\ref{eq:cost_fucntion}) and the potential function (\ref{eq:pot_cost_function}) that result in an identical feedback gain matrix. They can be additionally modified such that ${{\tilde{J}\play{i}} = \kappa^{i} \cdot  J\play{i}}$ and ${\tilde{J}\play{p}} = \kappa^{p} \cdot J\play{p}$, where ${\kappa^{p}, \kappa^{i}>0 \in \, {R}^n}$ without affecting the NE. This property is used to choose $\normTwo{{\vek{R}\play{p}}^{-1}}$ such that
		$ \normTwo{{\vek{R}\play{p}}^{-1}}\normTwo{{\vek{R}\play{i}}}\geq 1$ holds, which leads to the following upper estimation:
		\begin{align*}
			& \displaystyle \Delta \vek{E}  	\leq  \vek{B}\left( {\vek{B}}^T \vek{P}\play{p} - \sum_{i \in \mathcal{P}}  {\vek{B}\play{i}}^T \vek{P}\play{i} \right)  \\
			\hspace{3mm}& \leq \normTwo{\vek{B}}\normTwo{ {\vek{B}}^T \vek{P}\play{p} - \sum_{i \in \mathcal{P}}  {\vek{B}\play{i}}^T \vek{P}\play{i}}  \\
			\hspace{3mm}& \leq  \normTwo{\vek{B}} \cdot 2 \cdot  \underset{i}{\mathrm{max }} \normTwo{{\vek{B}}^{(i)} \vek{P}\play{p} -  {\vek{B}\play{i}}^T \vek{P}\play{i}} \leq  2\normTwo{\vek{B}}  \Delta.
		\end{align*} 
		After the substitution of $\Delta \vek{E}$, the proof of the theorem is obtained
		\begin{align} \label{eq:end_of_proof}
			\varepsilon(\Delta) =  \normTwo{e^{\vek{A}\play{p}_c \cdot t}\sv{x}_0 \left(e^{\left\lVert B\right\rVert_2 \cdot 2 \cdot \Delta \cdot t} - 1\right)}, \; \forall t \in [t_0,t_1].
		\end{align}
	\end{proof}
	\begin{remark}
		Lemma \ref{lem:npdg_boundedness} provides the link between the upper estimation of the trajectory errors $\varepsilon$ and the DD of the two games, cf. (\ref{eq:estim_delta}), which therefore differs from Gr\"onwall–Bellman inequality \cite{2015_NonlinearControl_khalil} 
	\end{remark}
	\begin{remark}
		Lemma \ref{lem:npdg_boundedness} states that the distance $\Delta$ holds ${\forall t \in [t_0,t_1]}$. In the following, we use the notation $\Delta_{t_i}$ for distances with respect to the timespan ${[t_i,t_{i+1}]}$. Assuming that
		$ \Delta_{t_i} \geq \Delta_{t_{i+1}} \forall i $, Lemma \ref{lem:npdg_boundedness} can be extended for $t \rightarrow \infty$. In case of asymptotic stable system dynamics of solutions $\sv{x}\play{p}(t)$ and $\sv{x}^*(t)$ is this assumption is feasible.
	\end{remark}
	\vspace*{-1mm}	
	\subsection{Method for Finding a LQ NPDG} \label{sec:optim_find_potgame}
	For practical applications of the presented NPDG, an identification method is provided to find an NPDG for a given differential game. First, the feedback control law of the players $ {\hat{\sv{K}} = \begin{bmatrix} \hat{\vek{K}}\play{\mathrm{a}},\hat{\vek{K}}\play{\mathrm{h}} \end{bmatrix}}$ are estimated from the measurements of $\sv{x}^*$ and $\sv{u}\play{\mathrm{a}}$ and $\sv{u}\play{\mathrm{h}}$.
	Then, an inverse optimization problem is formulated as a linear matrix inequality (LMI). We adapt method of the one-player inverse optimization problem from~\cite{2015_SolutionsInverseLQR_priess} for the identification of NPDGs. The minimization of the condition number $\beta$ is subject to LMI constraints, which provides a unique solution of the identification:
	\begin{subequations} \label{eq:optim_find_pot_games}
		\begin{align} 
			&\; \; \; \hat{\vek{Q}}\play{p}, \hat{\vek{R}}\play{p},\hat{\vek{P}}\play{p}, \hat{\beta} =  \underset{\vek{Q}\play{p},\vek{R}\play{p}, \vek{P}\play{p}, \beta}{\text{arg min }}  \beta ^2  \\ 
			\text{s.t. }& \vek{A}^T \vek{P}\play{p} +  \vek{P}\play{p} \vek{A} + \vek{Q}\play{p} - \vek{P}\play{p}\vek{B}\hat{\sv{K}} = \sv{0}, \\
			& \vek{B}^T\vek{P}\play{p} - \vek{R}\play{p}\hat{\sv{K}} = \sv{0},\\
			&\vek{I} \leq \begin{bmatrix}
				\vek{Q}\play{p} & \vek{0} \\
				\vek{0} & \vek{R}\play{p}
			\end{bmatrix} \leq \beta \vek{I}, \\
			&\underset{i}{\mathrm{max }}\left\lVert{\vek{B}\play{i}}^T\vek{P}\play{p} - {\vek{B}\play{i}}^T\vek{P}\play{i} \right\rVert_2 \cdot \sv{x}_{\mathrm{max}} < \Delta, \\
			&\normTwo{{\vek{R}\play{p}}^{-1}}\normTwo{{\vek{R}\play{i}}}\geq 1 \; \forall i \in \mathcal{P}.
		\end{align}
	\end{subequations}
	The matrix $\vek{P}\play{i}$ is the solution of (\ref{eq:sol_i_Ric}) computed from the original game. For the computation of (\ref{eq:optim_find_pot_games}), it is assumed that the system matrices {\Large(}$\vek{A}$ and $\vek{B}\play{i}${\Large)} and the cost function of the automation and the human ($\vek{Q}\play{i}, \vek{R}\play{i}, \; i \in \mathcal{P}$) are given. The value of $\Delta$ is iteratively chosen based on the maximal DD between the two games, cf. (\ref{eq:estim_delta}). 
	\vspace*{-1mm}
	\section{Systematic Design of LISC by means of NPDG} \label{sec:sys_derivation}
	In this section, an approach for the systematic design of the CS and the LISC design steps are presented.
	\vspace*{-1mm} 
	\subsection{Design of the CS}
	\vspace*{-1mm}
	The idea is the following: The cooperation is modelled as an NPDG, from which the matrices $\bm{\Xi}_{m}\play{\mathrm{a}}$ and $\bm{\Xi}_{m}\play{\mathrm{h}}$ of the CS (\ref{eq:linCooperation_state}) can be computed.
	For the sake of readability, the index ${(p)}$ symbolizing the NPDG and the time dependencies are omitted in the following. Let the Hamiltonian function of the NPDG be given as in (\ref{eq:potential_hamil_quad_def}) with linear system dynamics (\ref{eq:lin_system}). To find the NE, an optimal control problem for the potential function (\ref{eq:pot_cost_function}) with the underlying linear system dynamics (\ref{eq:lin_system}) has to be solved, which is more simple then solving the coupled Ricatti Equation (\ref{eq:sol_i_Ric}). The optimality conditions (see e.g.~\cite{2015_OptimierungStatischeDynamische_papageorgiu}) of the PG are
	\begin{subequations} \label{eq:Optimality_eq}
		\begin{align}
			\frac{\partial H }{\partial \sv{u} } &= \mathbf{R} \sv{u}  + \mathbf{B}  \sv{\lambda}  = \sv{0},\\
			\dot{\sv{x}} &= \frac{\partial H}{\partial \sv{\lambda}} = \mathbf{A} \sv{x} + \mathbf{B} \sv{u},\\
			\dot{\sv{\lambda}} &= -\frac{\partial H}{\partial \sv{x}} = -\mathbf{Q} \sv{x} + \mathbf{A}^T \sv{\lambda}.
		\end{align}
	\end{subequations}
	With the optimality conditions (\ref{eq:Optimality_eq}) the optimum of (\ref{eq:pot_cost_function}) is computed (see e.g~\cite{2015_OptimierungStatischeDynamische_papageorgiu}), in which it is assumed that the co-state $\lambda$ is a linear function of $\sv{x}$
	\begin{equation} \label{eq:way_B_1}
		\sv{\lambda} = \mathbf{P}\sv{x}
	\end{equation}
	For time invariant control law ($\dot{\mathbf{P}} = 0$) the time-derivative of the co-state is
	\begin{equation} \label{eq:way_B_2}
		\dot{\sv{\lambda}} = \mathbf{P}\dot{\sv{x}}.
	\end{equation}
	Substituting (\ref{eq:way_B_1}) and (\ref{eq:way_B_2}) in (\ref{eq:Optimality_eq}c) leads to ${\mathbf{P}\dot{\sv{x}} = -\mathbf{Q} \sv{x} + \mathbf{A}^T \mathbf{P}\sv{x},}$ for which the system dynamics (\ref{eq:Optimality_eq}b) can be applied $\mathbf{P}\left(\mathbf{A} \sv{x} + \mathbf{B} \sv{u }\right) = -\mathbf{Q} \sv{x} + \mathbf{A}^T \mathbf{P}\sv{x},$ which can be rearranged to express $\sv{x}$ as the function of input $\sv{u}$
	\begin{equation} \label{eq:coop_state_B}
		\sv{x} = -\left[\mathbf{P}\mathbf{A} + \mathbf{Q} - \mathbf{A}^T\mathbf{P}\right]^\dag \mathbf{P}\mathbf{B}\sv{u},
	\end{equation}
	where the index $\dag$ is Moore–Penrose inverse of a matrix, which is computed as $\vek{G}^\dag = ({\vek{G}^T\vek{G}})^{-1}\vek{G}^T$. Using (\ref{eq:coop_state_B}) the weights of the linear CS in (\ref{eq:linCooperation_state}) can be systematically derived aiming a better understanding of the CS. A splitting of the expression $-\left[\mathbf{P}\mathbf{A} + \mathbf{Q} - \mathbf{A}^T\mathbf{P}\right]^\dag \mathbf{P}\mathbf{B}$ into measurable and non-measurable states leads to {${\sv{x}_{nm} = \bm{\Xi}\play{\mathrm{a}} \sv{u}\play{\mathrm{a}} +\bm{\Xi}\play{\mathrm{h}} \sv{u}\play{\mathrm{h}}}$}, which is in the form given in~(\ref{eq:linCooperation_state}), if ${\sv{x}_\kappa := \sv{x}_{nm}}$ is applied. The matrices $\bm{\Xi}\play{\mathrm{a}}$ and $\bm{\Xi}\play{\mathrm{h}}$ are computed from the corresponding elements of (\ref{eq:coop_state_B}).	
	\subsection{Design of a LISC} \label{subsec:LISC}
	The systematic design of the LISC includes the following steps:
	
	\subsubsection*{1} Design of a FISC as defined using the following coupled optimization problem presented in \cite{2014_NecessarySufficientConditions_flad}:
	\begin{subequations} \label{eq:optim_flad}
		\begin{gather}
			\sv{\theta}\play{\mathrm{a}} = \underset{\sv{\theta}}{\text{arg min }} J\play{g}\\ 
			\begin{align} 
				\text{s.t. }&\sv{u}\play{i} = \underset{\sv{u}\play{i}}{\text{arg min }} J\play{i}, \  \  \ i \in \mathcal{P}, \\
				& \dot{\sv{x}}(t) = \vek{A} \sv{x}(t) + \sum_{i \in \mathcal{P}} \vek{B}\play{i}\sv{u}\play{i}(t), \ \ t \in [0,\tau]
			\end{align}
		\end{gather}
	\end{subequations}
	where \( \sv{\theta} = \left[ \vek{Q}\play{\mathrm{a}}_\text{FI}, \; \vek{R}\play{\mathrm{a}}_\text{FI}\right]\) denotes the parametrization of the FISC. $J\play{g}$ is the global goal function given by the system requirements.
	The optimization (\ref{eq:optim_flad}) yields the parameters of FISC. The constraints (\ref{eq:optim_flad}b) and (\ref{eq:optim_flad}c) ensure that the resulting FISC lead to a NE the game of the of the human-automation interaction. The optimization (\ref{eq:optim_flad}) is carried out with a sequential quadratic programming solver, see e.g.~\cite{2000_NumericalOptimization_nocedal}. 
	\subsubsection*{2} From the FISC an NPDG is identified with the use of the algorithm presented in Section \ref{sec:optim_find_potgame}.
	\subsubsection*{3} From the NPDG, the CS is derived according to Section~\ref{sec:sys_derivation}.
	\subsubsection*{4} Design of the LISC with a LQR-design, which has the cost function (\ref{eq:extended_cost_function}). The goal is that the LISC reaches a result similar to the FISC. Finding this desired behaviour happens with the optimization of a quadratic cost function (\ref{eq:extended_cost_function}) for the extended system (\ref{eq:extended_system_eqaution}), where the $\vek{Q}\play{\mathrm{a}}_\text{LI}$ and $\vek{R}\play{\mathrm{a}}_\text{LI}$ are computed with the nested optimization
	\begin{subequations} \label{eq:optim_varga_k}
		\begin{align}
			\vek{Q}\play{\mathrm{a}}_\text{LI}, \vek{R}\play{\mathrm{a}}_\text{LI} \; = & \underset{   \hat{\vek{Q}}\play{\mathrm{a}}_\text{LI},\hat{\vek{R}}\play{\mathrm{a}}_\text{LI} }{\text{arg min }} \, \normTwo{\sv{u}\play{\mathrm{a}}_\text{FI}(t) - \sv{u}\play{\mathrm{a}}_\text{LI}(t)} \\
			\text{s.t. } &\; (\ref{eq:extended_system_eqaution}), \ (\ref{eq:extended_cost_function}) \text{ and } (\ref{eq:optim_flad}).
		\end{align}
	\end{subequations}
	The constraints (\ref{eq:optim_varga_k}b) ensure that 
	\begin{itemize}
		\item the global goal $J\play{g}$ is fulfilled,
		\item the automation cost function $J\play{a}_{LI}$ is optimal and
		\item the players of the game, obtained from (\ref{eq:optim_flad}), optimize their original cost functions.
	\end{itemize}
	
	\section{Verification of the Design} \label{sec:verif_sim}
	This section presents the overall design procedure on a real world application of the large vehicle-manipulators (VM), for which a heuristic controller was presented earlier in \cite{2019_ControlLargeVehicleManipulators_varga}
	\subsection{Model of the vehicle-manipulator}
	Large VMs are used for example in roadside maintenance, where it is crucial to ensure a fast and safe work with these machines, see~\cite{2020_Produktkatalog_fiedler, 2020_ProduktubersichtArbeitsgerate_mulag}. Fig.~\ref{fig:frenet_frame} shows an example of a large VM. The VM is modelled in a Fr\'enet Frame, relative to the reference trajectories, cf. Section \ref{sec:concept_cooperation_sattes}. The VM moves with the longitudinal velocity $v$ and is controlled by steering angle $u\play{a} = \delta$ and by a rate controller for the length and the orientation of the manipulator $\sv{u}\play{h} = {[\dot{a}, \, \dot{\alpha}]^T}$. The system states are ${\sv{x} = \left[d_m, \; \Delta \alpha, \; d_v, \; \Delta \theta \right]}$, where $d_m$ and $\Delta \alpha$ are the lateral and the orientation error of the manipulator. Meanwhile, $d_v$ and $\Delta \theta$ are the lateral and the orientation error of the vehicle.
	\begin{figure}[b!]
		\centering
		
		\includegraphics[width=0.75\linewidth]{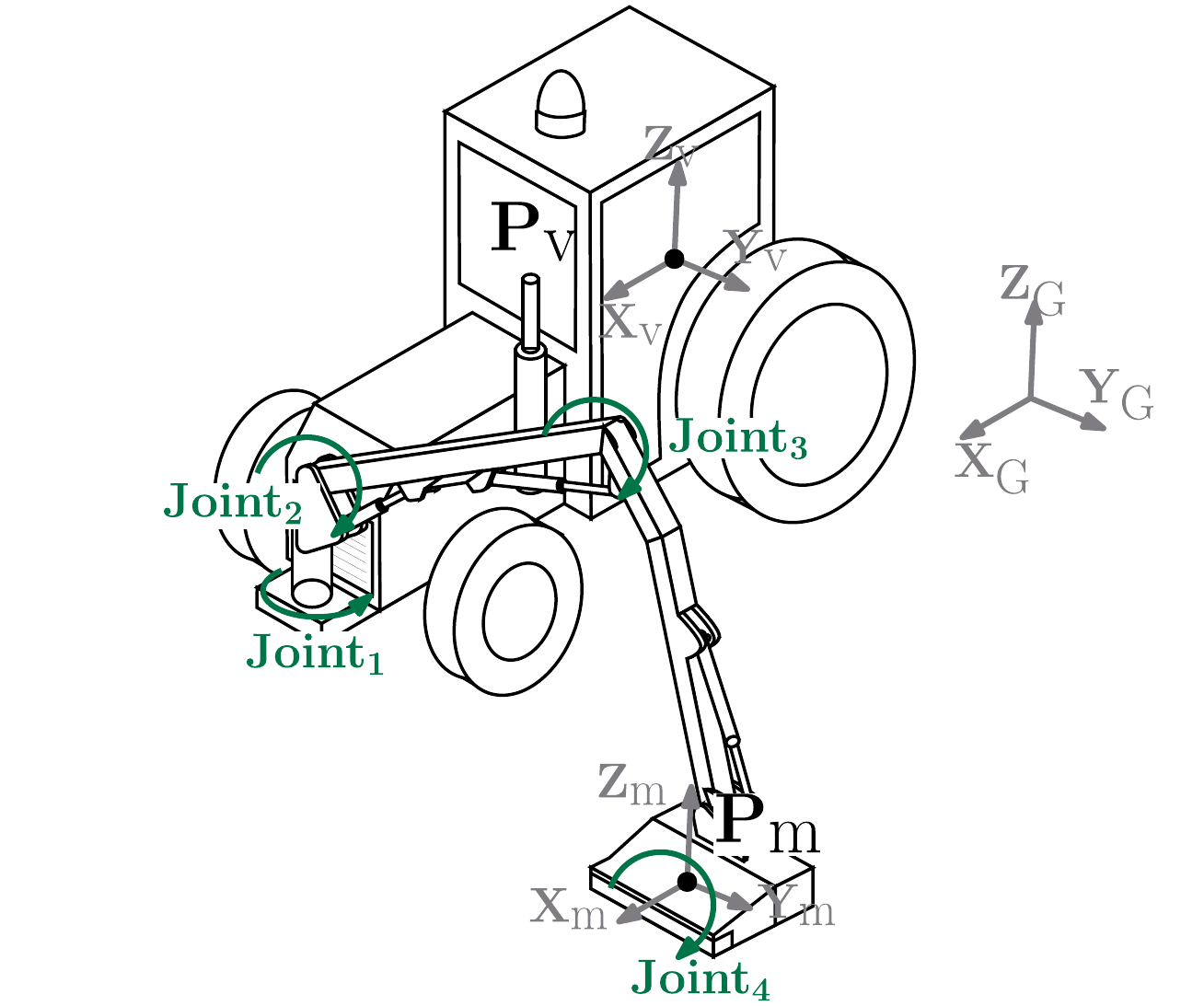}
		
		\caption{An example of a vehicle-manipulator \cite{2019_ModelPredictiveControl_varga}} 
		\label{fig:frenet_frame}
	\end{figure} 
	For a constant velocity $v$, a LTI system is obtained with the system and input matrices
	\begin{align*}
		A = \begin{bmatrix}
			0 & 0 & 0 & 0 \\
			0 & -1 & 0 & 0 \\
			0 & 0 & v & 0\\
			0 & 0 & 0 & 0 
		\end{bmatrix} & \\
		B\play{\mathrm{h}} \!  = \! \begin{bmatrix}
			\sin \alpha_e \!& 0 \!& 0 \!& 0 \\ 
			\mathrm{a}_e\cos \alpha_e \!& 1 \!& 0 \!& 0 
		\end{bmatrix}^T, \;	B\play{\mathrm{a}} & \! = \! \begin{bmatrix}
			L\cdot v \!&	0 \!&	0 \!&	v
		\end{bmatrix}^T\!,
	\end{align*}
	where $\alpha_e$ is the reference orientation angle of the manipulator, $\mathrm{a}_e$ is the length of the manipulator and $L$ is the distance between front and rear axle. The parameters are obtained from~\cite{2019_ModelPredictiveControl_varga}.	
	Controlling VMs can happen also by control methods from the literature, see e.g. \cite{2014_VehicleManipulatorSystemsModeling_from, 2018_ControlBasedLinear_andaluz, 2019_TrajectoryControlMobile_ram, 2022_MotionPlanningMobile_sandakalum}. However, in our application, it is not possible because 1) The automation and the human have to share the control tasks and 2) The manipulator errors are assumed not to be measurable.
	\subsection{Ground Truth Human Model} \label{sec:param_hum_model}
	
	\begin{figure*}[!hb] \label{fig:pot_ham}
		\setcounter{figure}{3}
		\centering
		\includegraphics[width=0.99\linewidth]{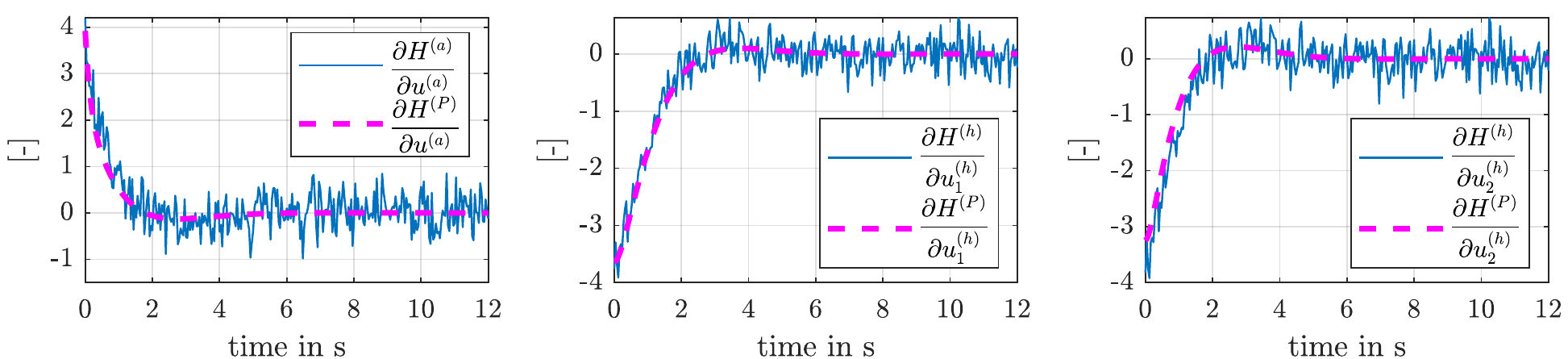}		
		\caption{The dynamics of Hamiltonian functions, the blue solid lines are the right side of (\ref{eq:def_e_pot})~with SNR$=5\,$dB and the purple dashed lines are the left side of (\ref{eq:def_e_pot}).}		
	\end{figure*}
	For the design, a ground truth of the human cost function is defined, which enables a comparison of the novel concepts with the state-of-the-art. The quadratic cost function $J\play{\mathrm{h}}$ is assumed to be identified in advance, for instance with the method of \cite{2015_SolutionsInverseLQR_priess}. The penalty matrices of the system states and the inputs are chosen to 
	${	\vek{Q}\play{\mathrm{h}} = \text{diag}\left(4.5,\;1,\;0.5,\;0.5\right)}$ and
	${	\vek{R}\play{\mathrm{h}} = \text{diag}\left(0,\;1.05,\;0.9\right)}$. 
	These penalty matrices model the human operator's main objective of minimizing the manipulator states meanwhile the deviation of the vehicle from its reference trajectories are only a little taken into account. 
	\subsection{Design of FISC} \label{subsec:FISC}
	A quadratic global cost function, 
	\begin{equation} \label{eq:glob_cost_w_param}
		J\play{g} = \frac{1}{2}\int_{0}^{\tau} \sv{x}^T \vek{Q}\play{g}\sv{x} +  \sv{u}^T\vek{R}\play{g}\sv{u} \text{ d}t,
	\end{equation}
	is used for the controller design, as suggested in~\cite{2014_NecessarySufficientConditions_flad}, where the weights are chosen to~$ \vek{R}\play{\text{g}} =   \text{diag}\left(1, \,1.45, \,1.35\right)$ and $ \vek{Q}\play{\text{g}} =  \text{diag}\left(5.5, \,0.5, \,1.25, \,0.85 \right) $. A support for the operator is ensured by the suitable choice of the penalty weights: The penalty gains of the lateral errors {\Large(}$\vek{Q}\play{\text{g}}_1$ and $\vek{Q}\play{\text{g}}_3${\Large)} for the manipulator's errors are larger than for the vehicle. However, rapid changes of the vehicle's orientation can be frustrating and counter-intuitive for the operator in contrast to the orientation error of the manipulator. For this reason, $\vek{Q}\play{\text{g}}_4$ is chosen larger than $\vek{Q}\play{\text{g}}_2$.
	The cost function of the automation is computed with (\ref{eq:optim_flad}) and the matrices obtained are ${\vek{Q}\play{\mathrm{a}}_\text{FI} = \text{diag}\left(4.21, \,    3.37, \,    1.32, \,    0.35\right)}$ and ${\vek{R}\play{\mathrm{a}}_\text{FI} = \text{diag}\left(1,\,0,\,0\right).}$ The initial values of the optimization~(\ref{eq:optim_flad}) are chosen to ${\sv{\theta}_0 = [5, \,0.1,  \,1, \,0.9, \,0.1, \, 0, \, 0]}$.
	\subsection{Identification of the NPDG with noisy measurements}
	After the design of the FISC, an NPDG which models this cooperative setup is sought. Applying the method for an NPDG (cf. (\ref{eq:optim_find_pot_games})), a quadratic potential function (\ref{eq:pot_cost_function}) can be identified. The maximum DD in (\ref{eq:optim_find_pot_games}) is chosen to $\Delta=0.05$ and the identified matrices of the NPDG are:
	\begin{figure}[t!]
		\setcounter{figure}{2}
		\includegraphics[width=0.99\linewidth,height=5.75cm]{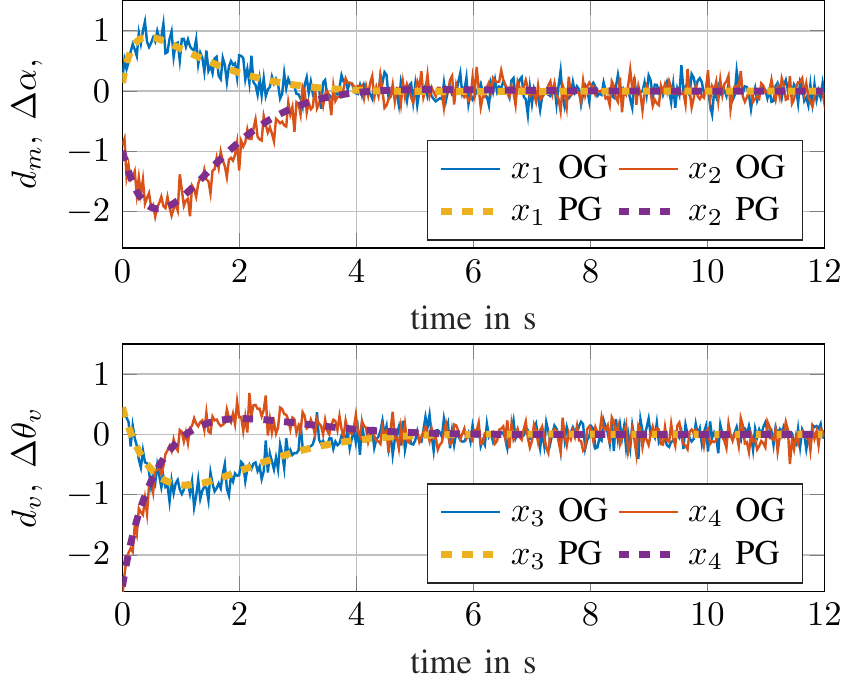}
		
		\caption{The resulting noisy trajectories of the original game (OG), with solid lines and the NPDG (PG) with dashed lines with SNR$=5\,$dB.}
		\label{fig:FISC_and_PG}
	\end{figure}
	\begin{equation*}
		\resizebox{0.999\hsize}{!}{$\mathbf{Q}\play{p} = 
			\begin{bmatrix}
				3.97&0.00&0.40&0.31\\
				0.00&1.17&0.00&0.00\\
				0.40&0.00&1.12&0.05\\
				0.31&0.00&0.05&0.38\\
			\end{bmatrix}
			\text{ and }
			\mathbf{R}\play{p} = 
			\begin{bmatrix}
				1.06&0.12&0.09\\
				0.12&1.11&0.01\\
				0.09&0.01&0.97\\
			\end{bmatrix}$}
	\end{equation*}
	The solution of the Ricatti equation is
	\begin{equation*}
		\mathbf{P}\play{p} = 
		\begin{bmatrix}
			0.988&-0.173&-0.439&-1.266\\
			-0.173&0.530&0.119&0.428\\
			-0.439&0.119&1.767&1.790\\
			-1.266&0.428&1.790&4.303\\
		\end{bmatrix}.
	\end{equation*}
	These above are the identified matrices without noisy signals. In order to analyse the robustness to measurement noise, white Gaussian noise is added to the trajectories analysing the properties of the NPDG as function of the noise level (SNR), $\tilde{\sv{x}}^*(t) = \sv{x}^*(t) + \sv{\varrho}(t) $. 
	For the measure of the deviation of the PG from the original game, we use $\mathrm{max}\,\sigma\play{i}_d(t) $, $\mathrm{max}\,(\normTwo{\sv{x}\play{p}(t)-\sv{x}^*(t)})$ and the $\Delta$ value that ensures the feasibility of (\ref{eq:optim_find_pot_games}).
	In Table \ref{table:res_traj_noise}, the results are given. It shows that 
	the smaller the SNR value, the greater the DD and consequently the value of $\Delta$ which has to be found for a feasible problem (\ref{eq:optim_find_pot_games}). Roughly speaking, this implies that the potential game is less similar to the original game with increasing noise.
	Still, the proposed algorithm can provide similar trajectories as the original game, see the trajectories with SNR$=5\,$dB on Fig. \ref{fig:FISC_and_PG}.  	
	\begin{table}[t!] 
		\small
		\caption{Results with different white Gaussian noise levels}
		\centering
		\begin{tabular}{c|ccc}
			SNR/dB & $\underset{i}{\mathrm{max}}\, \sigma\play{i}_d(t) $ & $\mathrm{max}\left(\normTwo{\sv{x}\play{p}(t)-\sv{x}(t)}\right)$ & $\Delta$ \\
			\hline
			5  & 0.1365& 0.0761  & 0.15\\
			10 & 0.1164& 0.0642 & 0.15\\
			20 & 0.0326& 0.0198 & 0.10 \\
			30 & 0.0010& 0.0123 & 0.05 \\
			\hline 
		\end{tabular}
		\label{table:res_traj_noise}
	\end{table}
\vspace*{1mm}
	\section{Validation on a Vehicle-Manipulator Test-Bench} \label{sec:validation}
	This section presents the test-bench, on which the study was carried out with sixteen test subjects. Furthermore, the analysis of the LISC designed with NPDGs is discussed.
	\begin{figure}[b!]
		\centering
		\setcounter{figure}{4}
		\includegraphics[width=0.72\linewidth]{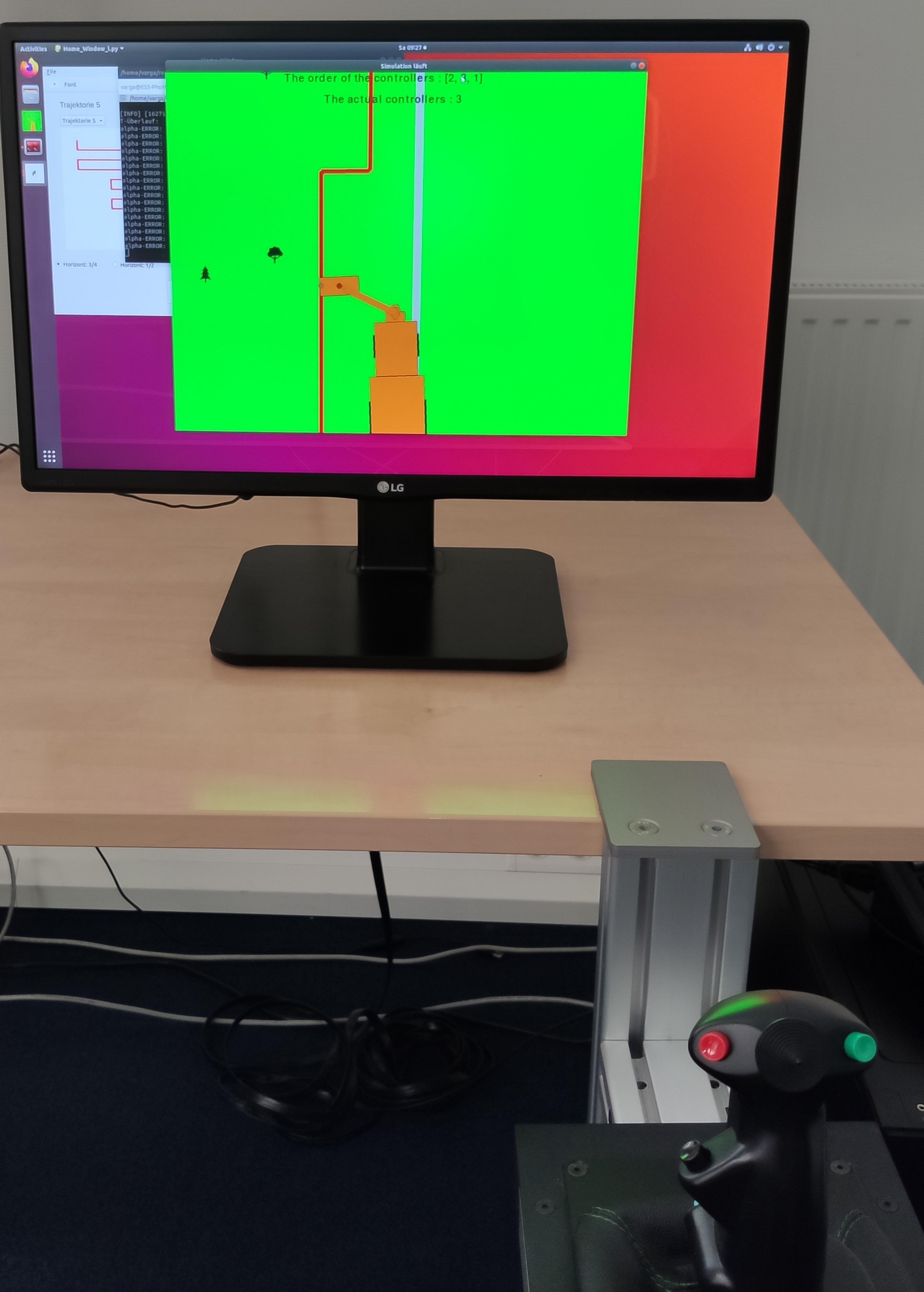}		
		\caption{Picture of the test bench with the GUI and the joystick}
		\label{fig:test_bench}
	\end{figure}
	\subsection{Test Bench Setup}
	The test bench can be seen in Fig.~\ref{fig:test_bench}. It consists of a simulation computer, a force feedback joystick from \textit{Brunner AG}, model \textit{CLS-E Brunner Jet} \cite{2022_BrunnerElektronikAG_} and a 2D GUI. The simulation was carried out an on Intel Core I7-5930K Processor at 3.5 GHz and with an Nvidia GeForce GTX 1070 GPU. The VM-model involves a three-dimensional physical model of a tractor with detailed suspension and tyre models, for details see \cite{2019_ModelPredictiveControl_varga}. The hydraulic manipulator is modelled as four rigid bodies connected by hydraulic cylinders, for which the detailed dynamics model of \cite{2017_FullReducedorderModel_ruderman} is used. The GUI is a 2D view of the VM and its environment. The components are implemented in the Robot Operating System framework. The software and hardware components are given in Fig.~\ref{fig:test_bench_software_struct}. To enable a comparison, the two reference trajectories of the manipulator (red line) and the vehicles (gray line) are given for the test subjects. In our experiment, there is no active force feedback on the joystick.
	
	\subsection{Experimental Protocol}
	Sixteen test subjects (4 female and 12 male, age $27.8\pm3.0$) took part in the experiment. All of them tested all the controllers, it is a within-subject experiment. The test subjects are not professional operators and have no prior experience with the simulator. They have the task of controlling the manipulator with the joystick to follow the red line as good as they can. In addition, they are asker to evaluate three different controller concepts:
	\begin{itemize}
		\item Non-cooperative (NC) controller, which controls the vehicle without the consideration of the manipulator. This concept is taken into account as the state-of-the-art of large VMs.
		\item FISC as given in Section \ref{subsec:FISC}. The FISC is used as the \textit{ideal solution}, which is the baseline for the required behaviour of the LISC. FISC is not applicable in general scenarios.  
		\item The proposed LISC, with the design through the NPDG as given in Section \ref{subsec:LISC}.
	\end{itemize} 
	The controllers are tested in a randomized order and the test subjects are unaware of which controller they are currently testing.
	The experimental protocol starts with a familiarization process: The test subjects have the possibility to get familiar with the control of the manipulator. This part took approximately~$250\,$sec. They are allowed to do anything to learn to control the manipulator. This is followed by the actual scenario with two typical types of trajectories (sudden step forms and smoother V-forms) and with the velocity of the vehicle set to $v=1.2\,$m/s, which is a common speed for roadside working with a large VM, see e.g.~\cite{2020_Produktkatalog_fiedler, 2020_ProduktubersichtArbeitsgerate_mulag}. The independent variable is the choice of the controller (FISC, NC, LISC). These runs took approximately $700\,$sec. Between these runs, the test subjects had the possibility to take some notes about the controllers.
	Finally, they have to evaluate the controllers by answering three questions, see \ref{sec:validation}-D.
	
	\subsection{Controller Setup}
	The controllers have the main task to keep the vehicle on its reference path. Furthermore, they should help the operator to carry out the task with the manipulator better.
	\begin{figure}[t!]
		\centering
		\setcounter{figure}{5}
		\includegraphics[width=0.9\linewidth]{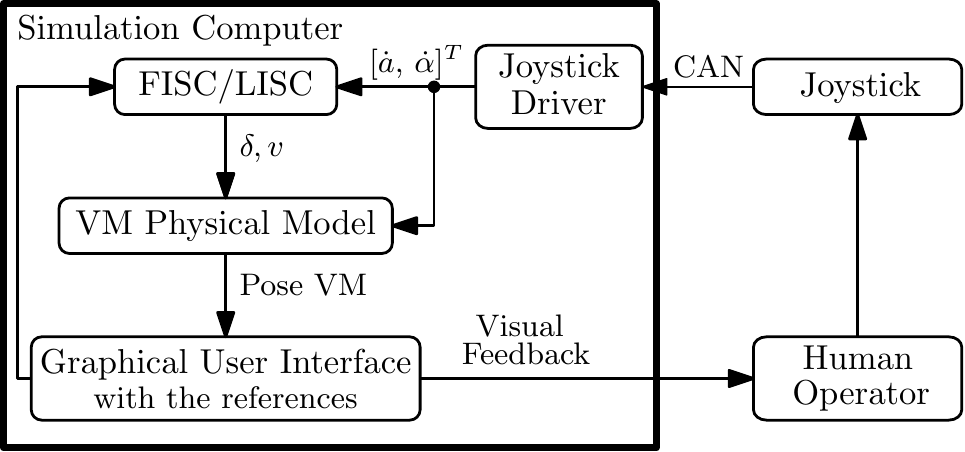}		
		\caption{Components of the test bench}
		\label{fig:test_bench_software_struct}
	\end{figure}
	\subsubsection{Non-cooperative controller (NC)} 
	The non-cooperative controller (NC) controls only the vehicle states, $d_v$ and $\Delta \theta$. This is taken into account as the simplest state-of-the-art solution of the vehicle control without a consideration of the manipulator. Similar controllers for autonomous vehicles can be found in the state-of-the-art, see e.g. \cite{2010_OptimalTrajectoryGeneration_werling, 2011_FullyAutonomousDriving_levinson, 2016_ReviewMotionPlanning_gonzalez}. The controller used in the experiment is ${u_\mathrm{NC}\play{\mathrm{a}} = -\sv{K}_\mathrm{NC} \cdot [d, \, \Delta \theta_v]}$, where feedback gains obtained from a LQR-controller design are 
	$\sv{K}_\mathrm{NC} = [1.1, 3.2].$
	\subsubsection{FISC (FI)}
	The FISC is designed with the human model from Section~\ref{sec:param_hum_model} and with the global cost function~(\ref{eq:glob_cost_w_param}). The optimization (\ref{eq:optim_flad}) is carried out and the feedback gain obtained from the optimizations is
	$\sv{K}_\text{FI} = [1.02, \, -0.06,\, 0.65,\, 1.31].$
	\subsubsection{LISC (LI)}
	The LISC is designed with the procedure given in \ref{subsec:LISC}. A linear CS is used as given in Definition 1, where the parameters are obtained from (\ref{eq:linCooperation_state}). The result is
	\begin{align*}
		\Xi_{m}^{(h)} = \begin{bmatrix}
			0.144&0.192\\
			0.055&0.454\\
		\end{bmatrix} \text{ and }\;
		\Xi_{m}^{(a)} = \begin{bmatrix}
			0.470\\
			-0.127\\
		\end{bmatrix}.
	\end{align*}
	The feedback gain for the extended system dynamics (\ref{eq:extended_system_eqaution}) is
	$$ \sv{K}_\text{LI} = [11.11, \; 25.32, \; 30.50, \; -29.01, \; -22.31],$$
	which is obtained from optimization of (\ref{eq:extended_cost_function}) with parameters computed from (\ref{eq:optim_varga_k}). 
	\subsection{Evaluation Criteria}
	The goals of the experiment are threefold: First, the usability of the design with NPDG is validated. Then, the difference between the state-of-the-art NC controller and the novel LISC concept is analysed statistically. This comparison has a practical relevance: Both NC and LISC may be implemented on a real VM with similar sensors and hardware setup. Finally, the differences between the FISC and the LISC are investigated. 
	
	For the objective evaluation, a stack consisting of M data points is collected from the measurements with $25\,Hz$. Additionally, the performance, defined as the root mean square error (RMSE) of the manipulator
	\begin{equation}
		|d_{m}| = \sqrt{\frac{1}{M}\sum_{k=1}^{M} d^2_m[k]},
	\end{equation}
	is computed for the evaluation. The number of collected data points is $M \approx 2490.$
	The subjective evaluation of the controllers is made by means of three questions:
	\begin{itemize}
		\item[Q1] How do you assess your task performance? \\
		(Insufficient 0 -- 10 very good)
		\item[Q2] How intuitive did you find the support? \\
		(Not intuitive 0 -- 10 very intuitive)
		\item[Q3] How useful was the support to better accomplish the task? \\
		(Very disturbing 0 -- 10 very helpful)
	\end{itemize}
	The experimental data is used for the analysis of the following two hypotheses:
	\begin{itemize}
		\item [H1] The use of the LISC leads to a significant task performance improvement compared to a NC.
		\item [H2] There is no significant task performance difference between the FISC and the LISC using the design with the NPDG.
	\end{itemize}
	Note that the three controllers are not compared to each other at once and that the comparisons are not classical pairwise ones: H1 is a difference testing of LISC and NC. This is done with one-tailed, paired-sample t-tests \cite{2014_DataAnalysis_brandt} with the significance level of $\alpha_\text{H1} = 0.02$. An additional Bonferroni correction is applied due to the use of the data for H2, therefore, the corrected significance level is $\tilde{\alpha}_\text{H1} = \frac{\alpha_\text{H1}}{2} = 0.01$. On the other hand, the second hypothesis is an equivalence testing in accordance to \cite{1993_UsingSignificanceTests_rogers}. This happens by two one-sided t-tests (TOST) with $90\,$\% confidence interval for the difference. For more details of TOST, we referr to \cite{2005_TtestStatisticalEquivalence_limentani} and \cite{2018_EquivalenceTestingPsychological_lakens}.
	\section{Results and Discussion}
	\subsection{Objective Results}
	Table \ref{table:dev_manip_avg} shows the mean values and the standard deviation of the average errors of the manipulator $|d_m|$\footnote{For the sake of brevity in the discussion of the results, the indexes FI and LI are used instead of FISC and LISC, respectively.}. As there are no significant difference between the two runs with the different trajectory forms, therefore, they are analysed together, cf. Subsection \ref{sec:validation}-B. The results show that the NC has the weakest performance. For the statistical test of H1, LI and NC are compared: The p-value is ${p_\text{LI-NC} = 7.1\cdot10^{-3}}$, which shows that LI is significantly better than NC. Thus $\tilde{\alpha}_\text{H1}>p_\text{LI-NC}$ holds, H1 is accepted. 
	
	For the equivalence testing of H2, FI and LI are analysed by a TOST, which provides two p-values for both sides: If one of the values is greater then the significance level (${\alpha_\text{H2} = 0.05}$), there is no statistical difference between the mean values of FI and LI. The p-values are ${p_\text{FI-LI} = [0.15, \; 0.43]}$, therefore, H2 is also accepted.
	\begin{table}[b!] 
		\small 
		\centering
		\caption{The mean value of RMSEs of the manipulator from the reference $\mu_{\left|d_{m}\right|}$ and their standard deviation $\sigma_{\left|d_m\right|}$}
		\begin{tabular}{cccc}
			&NC & FI & LI \\
			\hline
			$\mu_{\left|d_{m}\right|}$ in m& 0.502 & 0.327 &  0.349 \\
			$\sigma_{\left|d_m\right|}$ in m& 0.264 & 0.216  & 0.219 \\
			\hline
		\end{tabular} 
		\label{table:dev_manip_avg}
	\end{table}
	\vspace*{-1mm}
	\subsection{Subjective Results}
	The results of the questionnaire enhance the results of the quantitative results, given in Table~\ref{table:questions}. The resulting p-values of the one-sided t-test are 
	$
	p^{Q1}_{\text{LI-NC}} = 3.2\cdot10^{-3},$ $p^{Q2}_{\text{LI-NC}} = 1.4\cdot10^{-5}$ and $p^{Q3}_{\text{LI-NC}} = 9.5\cdot10^{-6},$ showing that LI is significantly better than NC: LI leads to better subjective assessment of the task performance, is more intuitive and is more useful compared to NC.
	
	The results of FI and LI are not significantly different. Proving this equivalence between them, for each question a TOST is applied, with a $90\,$\% confidence interval. The resulting p-value pairs are
	$p^{Q1}_{\text{LI-FI}} = [0.29, \,2.16\cdot10^{-4}],$ $p^{Q2}_{\text{LI-FI}} = [0.22,\,2.03\cdot10^{-3}]$ and $ p^{Q3}_{\text{LI-FI}} = [0.36, \, 6.4\cdot10^{-4}],$
	which show that there is no statistical difference between the proposed LI and the FI.
	
	The verbal feedback of the test subjects provided similar results: The difference between FI and LI is not noticeable for the test subjects. The third test subject expressed his personal point of view by "Controller 1 [LI] is more aggressively configured than controller 2 [FI]" or test subject 11 said: "Controller 3 [FI] helps more in small curves than controller 1 [LI], but not in large curves". On the other hand, all the test subjects are able to distinguish NC from FI and LI. 
	The analysis of the subjective assessment supports the acceptance of H1 and H2.
	\begin{table}[b!] 
		\small
		\caption{Mean values and the standard deviations (in the brackets) of the personal questionnaire results}
		\centering
		\begin{tabular}{l|ccc}
			&NC & FI & LI \\
			\hline
			Q1-Self-assessment        & 4.94 (2.21) & 7.50 (1.10) & 6.75 (1.39)\\ 
			Q2-Intuition              & 4.00 (2.10) & 7.81 (1.42) & 7.21 (1.47)\\
			Q3-Support helpfulness    & 3.50 (2.28) & 7.81 (1.47) & 7.00 (1.41)\\
			\hline 
		\end{tabular}
		\label{table:questions}
	\end{table}
	\begin{figure}[t!]
		\centering	
		\includegraphics[width=\linewidth,height=5.95cm]{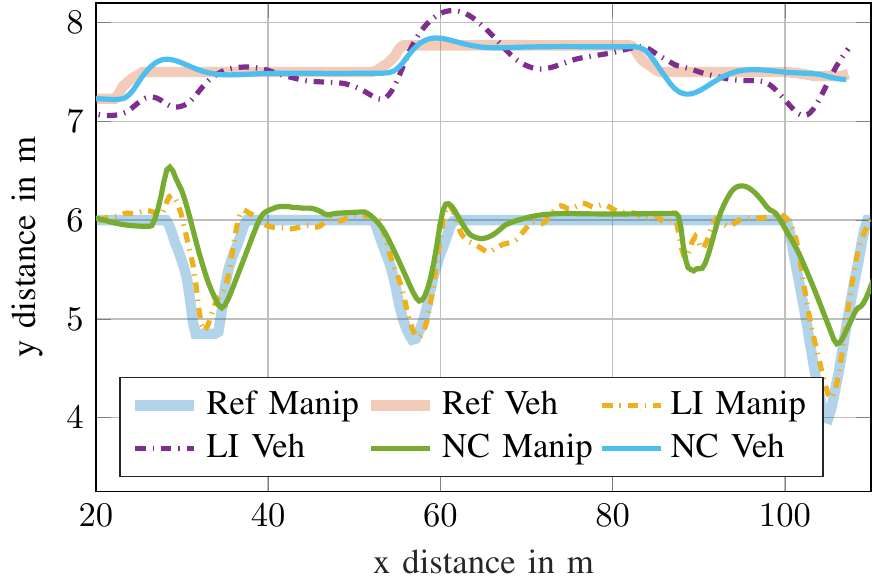}
		\caption{Comparison of the overall performance (test subject 6) to track the references (thick lines) of vehicle and manipulator	using a controller with no cooperative support (thin line) and a LI (dashed)}
	\label{fig:Meas_19_p4}
\end{figure}

\subsection{Discussion}
The resulting trajectories of the different concepts are presented in the following figures. In Fig.~\ref{fig:Meas_19_p4}, NC and LI are compared. It can be seen that with LI, a better tracking of the manipulator references is possible. Whereas, Fig. \ref{fig:Meas_125_p4_ful_lim} shows a comparison between FI and LI. It can be seen that both the vehicle and the manipulator trajectories are similar. However, there are still some differences, for which a possible reason may be the predefined human model. This indicates that an identification of the human operator can be beneficial in future case studies. The benefit of the novel LI is clear from Fig. \ref{fig:Meas_125_p4_ful_lim}: Despite of the limited information from the references the LI can provides a similar support, which does not differ strongly from a FI. The practical benefit is the saving of the necessary sensor on the manipulator. 
\section{Conclusion and Outlook} \label{sec:conclu}
This paper presented a systematic design method of the limited information shared control. The design happens with a novel class of differential games, the near potential differential games. A near potential differential game can model a cooperative human-machine-interaction, which is used for the design. 
With a practical use-case, we verified in simulations the limited information shared control design with the concept of the near potential differential games. Furthermore, the limited information shared control is validated on a test-bench. A statistical analysis with sixteen test subjects is carried out. This shows that the novel design of the limited information shared control leads to a behaviour similar to the controller needs all the system states and trajectories.
In our future work, we plan to combine the longitudinal guidance of the vehicle-manipulator with the lateral controller.
\begin{figure}[t!]
\centering
\includegraphics[width=\linewidth,height=5.95cm]{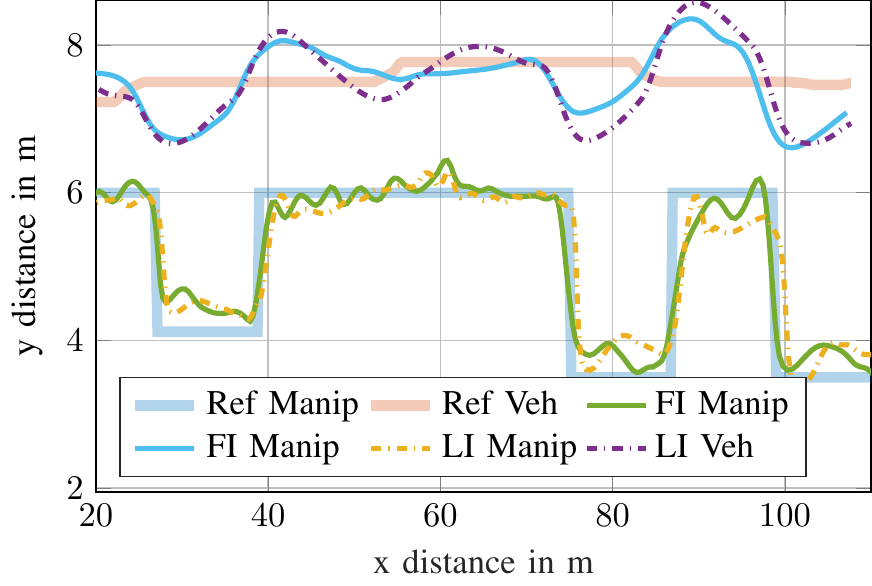}
\caption{Comparison of the overall performance (test subject 12) to track the references (thick lines) of vehicle and manipulator	using a controller with FI (thin line) and LI (dashed)}
\label{fig:Meas_125_p4_ful_lim}
\end{figure}
\bibliographystyle{IEEEtran}
\bibliography{bib_2022v6}

\begin{thebibliography}{10}
\providecommand{\url}[1]{#1}
\csname url@samestyle\endcsname
\providecommand{\newblock}{\relax}
\providecommand{\bibinfo}[2]{#2}
\providecommand{\BIBentrySTDinterwordspacing}{\spaceskip=0pt\relax}
\providecommand{\BIBentryALTinterwordstretchfactor}{4}
\providecommand{\BIBentryALTinterwordspacing}{\spaceskip=\fontdimen2\font plus
\BIBentryALTinterwordstretchfactor\fontdimen3\font minus
  \fontdimen4\font\relax}
\providecommand{\BIBforeignlanguage}[2]{{%
\expandafter\ifx\csname l@#1\endcsname\relax
\typeout{** WARNING: IEEEtran.bst: No hyphenation pattern has been}%
\typeout{** loaded for the language `#1'. Using the pattern for}%
\typeout{** the default language instead.}%
\else
\language=\csname l@#1\endcsname
\fi
#2}}
\providecommand{\BIBdecl}{\relax}
\BIBdecl

\bibitem{2018_TopologySharedControl_abbink}
D.~A. Abbink, T.~Carlson, M.~Mulder, J.~C.~F. {de Winter}, F.~Aminravan, T.~L.
  Gibo, and E.~R. Boer, ``A {{Topology}} of {{Shared Control
  Systems}}\textemdash{{Finding Common Ground}} in {{Diversity}},'' \emph{IEEE
  Transactions on Human-Machine Systems}, vol.~48, no.~5, pp. 509--525, Oct.
  2018.

\bibitem{2020_Shared_control_Eraslan}
E.~Eraslan, Y.~Yildiz, and A.~M. Annaswamy, ``Shared control between pilots and
  autopilots: {{An}} illustration of a cyberphysical human system,'' \emph{IEEE
  Control Systems Magazine}, vol.~40, no.~6, pp. 77--97, 2020.

\bibitem{2010_BlendedSharedControl_enes}
A.~Enes and W.~Book, ``Blended {{Shared Control}} of {{Zermelo}}'s navigation
  problem,'' in \emph{Proceedings of the 2010 {{American Control
  Conference}}}.\hskip 1em plus 0.5em minus 0.4em\relax {Baltimore, MD}:
  {IEEE}, Jun. 2010, pp. 4307--4312.

\bibitem{2013_HapticSupportBimanual_kuiper}
R.~J. Kuiper, J.~C. Frumau, F.~C. van~der Helm, and D.~A. Abbink, ``Haptic
  {{Support}} for {{Bi-manual Control}} of a {{Suspended Grab}} for {{Deep-Sea
  Excavation}},'' in \emph{2013 {{IEEE International Conference}} on
  {{Systems}}, {{Man}}, and {{Cybernetics}}}.\hskip 1em plus 0.5em minus
  0.4em\relax {Manchester}: {IEEE}, Oct. 2013, pp. 1822--1827.

\bibitem{2019_SharedControlSolution_devigne}
L.~Devigne, F.~Pasteau, T.~Carlson, and M.~Babel, ``A shared control solution
  for safe assisted power wheelchair navigation in an environment consisting of
  negative obstacles: A proof of concept,'' in \emph{2019 {{IEEE International
  Conference}} on {{Systems}}, {{Man}} and {{Cybernetics}} ({{SMC}})}.\hskip
  1em plus 0.5em minus 0.4em\relax {Bari, Italy}: {IEEE}, Oct. 2019, pp.
  1043--1048.

\bibitem{2010_NeuromuscularAnalysisGuideline_abbink}
D.~Abbink and M.~M., ``Neuromuscular {{Analysis}} as a {{Guideline}} in
  designing {{Shared Control}},'' in \emph{Advances in {{Haptics}}},
  M.~Hosseini, Ed.\hskip 1em plus 0.5em minus 0.4em\relax {InTech}, Apr. 2010.

\bibitem{2018_AdaptiveGameTheoreticDecision_tian}
R.~Tian, S.~Li, N.~Li, I.~Kolmanovsky, A.~Girard, and Y.~Yildiz, ``Adaptive
  {{Game-Theoretic Decision Making}} for {{Autonomous Vehicle Control}} at
  {{Roundabouts}},'' in \emph{2018 {{IEEE Conference}} on {{Decision}} and
  {{Control}} ({{CDC}})}.\hskip 1em plus 0.5em minus 0.4em\relax {Miami Beach,
  FL}: {IEEE}, Dec. 2018, pp. 321--326.

\bibitem{2019_HapticSharedSteering_zwaan}
H.~M. Zwaan, S.~M. Petermeijer, and D.~A. Abbink, ``Haptic shared steering
  control with an adaptive level of authority based on time-to-line crossing,''
  \emph{IFAC-PapersOnLine}, vol.~52, no.~19, pp. 49--54, 2019.

\bibitem{2020_ReviewSharedControl_marcano}
M.~Marcano, S.~Diaz, J.~Perez, and E.~Irigoyen, ``A {{Review}} of {{Shared
  Control}} for {{Automated Vehicles}}: {{Theory}} and {{Applications}},''
  \emph{IEEE Transactions on Human-Machine Systems}, vol.~50, no.~6, pp.
  475--491, Dec. 2020.

\bibitem{2011_OptimalVehicleStability_tamaddoni}
S.~H. Tamaddoni, M.~Ahmadian, and S.~Taheri, ``Optimal vehicle stability
  control design based on preview game theory concept,'' in \emph{Proceedings
  of the 2011 {{American Control Conference}}}.\hskip 1em plus 0.5em minus
  0.4em\relax {San Francisco, CA}: {IEEE}, Jun. 2011, pp. 5249--5254.

\bibitem{2014_NecessarySufficientConditions_flad}
M.~Flad, J.~Otten, S.~Schwab, and S.~Hohmann, ``Necessary and sufficient
  conditions for the design of cooperative shared control,'' in \emph{2014
  {{IEEE International Conference}} on {{Systems}}, {{Man}}, and
  {{Cybernetics}} ({{SMC}})}.\hskip 1em plus 0.5em minus 0.4em\relax {San
  Diego, CA, USA}: {IEEE}, Oct. 2014, pp. 1253--1259.

\bibitem{2015_GameTheoreticModelingSteering_na}
X.~Na and D.~J. Cole, ``Game-{{Theoretic Modeling}} of the {{Steering
  Interaction Between}} a {{Human Driver}} and a {{Vehicle Collision Avoidance
  Controller}},'' \emph{IEEE Transactions on Human-Machine Systems}, vol.~45,
  no.~1, pp. 25--38, Feb. 2015.

\bibitem{2019_ModellingOvertakingStrategy_farah}
H.~Farah, G.~Bianchi~Piccinini, M.~Itoh, and M.~Dozza, ``Modelling overtaking
  strategy and lateral distance in car-to-cyclist overtaking on rural roads:
  {{A}} driving simulator experiment,'' \emph{Transportation Research Part F:
  Traffic Psychology and Behaviour}, vol.~63, pp. 226--239, May 2019.

\bibitem{2002_OptimalFeedbackControl_todorov}
E.~Todorov and M.~I. Jordan, ``Optimal feedback control as a theory of motor
  coordination,'' \emph{Nature Neuroscience}, vol.~5, no.~11, pp. 1226--1235,
  Nov. 2002.

\bibitem{2009_NashEquilibriaMultiAgent_braun}
D.~A. Braun, P.~A. Ortega, and D.~M. Wolpert, ``Nash {{Equilibria}} in
  {{Multi-Agent Motor Interactions}},'' \emph{PLoS Computational Biology},
  vol.~5, no.~8, p. e1000468, Aug. 2009.

\bibitem{2014_AUVNavigationLocalization_paull}
L.~Paull, S.~Saeedi, M.~Seto, and H.~Li, ``{{AUV Navigation}} and
  {{Localization}}: {{A Review}},'' \emph{IEEE Journal of Oceanic Engineering},
  vol.~39, no.~1, pp. 131--149, Jan. 2014.

\bibitem{2011_OpenloopControlExperiments_moralesa}
D.~O. Morales, S.~Westerberg, P.~La~Hera, U.~Mettin, L.~B. Freidovich, and
  A.~S. Shiriaev, ``Open-loop control experiments on driver assistance for
  crane forestry machines,'' in \emph{2011 {{IEEE International Conference}} on
  {{Robotics}} and {{Automation}}}.\hskip 1em plus 0.5em minus 0.4em\relax
  {Shanghai, China}: {IEEE}, May 2011, pp. 1797--1802.

\bibitem{2019_ControlLargeVehicleManipulators_varga}
B.~Varga, A.~Shahirpour, S.~Schwab, and S.~Hohmann, ``Control of {{Large
  Vehicle-Manipulators}} with {{Human Operator}},'' \emph{IFAC-PapersOnLine},
  vol.~52, no.~30, pp. 373--378, 2019.

\bibitem{2020_LimitedInformationCooperativeShared_varga}
B.~Varga, A.~Shahirpour, M.~Lemmer, S.~Schwab, and S.~Hohmann,
  ``Limited-{{Information Cooperative Shared Control}} for
  {{Vehicle-Manipulators}},'' in \emph{{{IEEE International Conference}} on
  {{Systems}}, {{Man}}, and {{Cybernetics}} ({{IEEE SMC}} 2020)}.\hskip 1em
  plus 0.5em minus 0.4em\relax {IEEE, Piscataway, NJ}, 2020, p.~8.

\bibitem{1996_PotentialGames_monderer}
D.~Monderer and L.~S. Shapley, ``Potential {{Games}},'' \emph{Games and
  Economic Behavior}, vol.~14, no.~1, pp. 124--143, May 1996.

\bibitem{2013_NearPotentialGamesGeometry_candogan}
O.~Candogan, A.~Ozdaglar, and P.~A. Parrilo, ``Near-{{Potential Games}}:
  {{Geometry}} and {{Dynamics}},'' \emph{ACM Transactions on Economics and
  Computation}, vol.~1, no.~2, pp. 1--32, May 2013.

\bibitem{2017_SharedCooperativeControl_flemisch}
F.~Flemisch, Y.~Canpolat, E.~Altendorf, G.~Wesel, M.~Itoh, F.~Flemisch,
  M.~Baltzer, M.-P. {Pacaux-Lemoine}, D.~Abbink, and P.~Schutte, ``Shared and
  cooperative control of ground and air vehicles: {{Introduction}} and general
  overview,'' in \emph{2017 {{IEEE International Conference}} on {{Systems}},
  {{Man}}, and {{Cybernetics}} ({{SMC}})}.\hskip 1em plus 0.5em minus
  0.4em\relax {Banff, AB}: {IEEE}, Oct. 2017, pp. 858--863.

\bibitem{2015_SolutionsInverseLQR_priess}
M.~C. Priess, R.~Conway, J.~Choi, J.~M. Popovich, and C.~Radcliffe, ``Solutions
  to the {{Inverse LQR Problem With Application}} to {{Biological Systems
  Analysis}},'' \emph{IEEE Transactions on Control Systems Technology}, no.~2,
  2015.

\bibitem{2019_InverseDiscountedbasedLQR_el-hussieny}
H.~{El-Hussieny} and J.-H. Ryu, ``Inverse discounted-based {{LQR}} algorithm
  for learning human movement behaviors,'' \emph{Applied Intelligence},
  vol.~49, no.~4, pp. 1489--1501, Apr. 2019.

\bibitem{2005_LQDynamicOptimization_engwerda}
J.~Engwerda, \emph{{{LQ Dynamic Optimization}} and {{Differential Games}}},
  tilburg university, the netherlands~ed., 2005.

\bibitem{2016_SpringerHandbookRobotics_brunosiciliano}
{Bruno Siciliano} and O.~Khatib, \emph{Springer Handbook of Robotics},
  2nd~ed.\hskip 1em plus 0.5em minus 0.4em\relax {New York, NY}: {Springer
  Berlin Heidelberg}, 2016.

\bibitem{2020_ValidationCooperativeSharedControl_vargab}
B.~Varga, A.~Shahirpour, Y.~Burkhardt, S.~Schwab, and S.~Hohmann, ``Validation
  of {{Cooperative Shared-Control Concepts}} for {{Large
  Vehicle-Manipulators}},'' in \emph{2020 {{IEEE Conference}} on {{Control
  Technology}} and {{Applications}} ({{CCTA}})}.\hskip 1em plus 0.5em minus
  0.4em\relax {Montreal, QC, Canada}: {IEEE}, Aug. 2020, pp. 542--548.

\bibitem{2016_SurveyStaticDynamic_gonzalez-sanchez}
D.~{Gonz{\'a}lez-S{\'a}nchez} and O.~{Hern{\'a}ndez-Lerma}, ``A survey of
  static and dynamic potential games,'' \emph{Science China Mathematics},
  vol.~59, no.~11, pp. 2075--2102, Nov. 2016.

\bibitem{2018_PotentialDifferentialGames_fonseca-morales}
A.~{Fonseca-Morales} and O.~{Hern{\'a}ndez-Lerma}, ``Potential {{Differential
  Games}},'' \emph{Dynamic Games and Applications}, vol.~8, no.~2, pp.
  254--279, Jun. 2018.

\bibitem{2016_PotentialGameTheory_la}
Q.~D. L{\~a}, Y.~H. Chew, and B.-H. Soong, \emph{Potential {{Game
  Theory}}}.\hskip 1em plus 0.5em minus 0.4em\relax {Cham}: {Springer
  International Publishing}, 2016.

\bibitem{2015_NonlinearControl_khalil}
H.~K. Khalil, \emph{Nonlinear Control}.\hskip 1em plus 0.5em minus 0.4em\relax
  {Boston}: {Pearson}, 2015.

\bibitem{2015_OptimierungStatischeDynamische_papageorgiu}
M.~Papage{\=o}rgiu, M.~Leibold, and M.~Buss, \emph{{Optimierung: statische,
  dynamische, stochastische Verfahren f\"ur die Anwendung}}, 4th~ed.\hskip 1em
  plus 0.5em minus 0.4em\relax {Berlin Heidelberg}: {Springer Vieweg}, 2015.

\bibitem{2000_NumericalOptimization_nocedal}
J.~Nocedal and S.~J. Wright, \emph{Numerical Optimization}, corr. 2. print~ed.,
  ser. Springer Series in Operations Research.\hskip 1em plus 0.5em minus
  0.4em\relax {New York, NY}: {Springer}, 2000.

\bibitem{2020_Produktkatalog_fiedler}
{Fiedler}, ``Produktkatalog,'' Apr. 2020.

\bibitem{2020_ProduktubersichtArbeitsgerate_mulag}
{MULAG}, ``Produkt\"ubersicht: {{Arbeitsger\"ate}},'' Apr. 2020.

\bibitem{2019_ModelPredictiveControl_varga}
B.~Varga, S.~Meier, S.~Schwab, and S.~Hohmann, ``Model {{Predictive Control}}
  and {{Trajectory Optimization}} of {{Large Vehicle-Manipulators}},'' in
  \emph{2019 {{IEEE International Conference}} on {{Mechatronics}}
  ({{ICM}})}.\hskip 1em plus 0.5em minus 0.4em\relax {Ilmenau, Germany}:
  {IEEE}, Mar. 2019, pp. 60--66.

\bibitem{2014_VehicleManipulatorSystemsModeling_from}
P.~J. From, J.~T. Gravdahl, and K.~Y. Pettersen, \emph{Vehicle-{{Manipulator
  Systems}}: {{Modeling}} for {{Simulation}}, {{Analysis}}, and {{Control}}},
  ser. Advances in {{Industrial Control}}.\hskip 1em plus 0.5em minus
  0.4em\relax {London}: {Springer London}, 2014.

\bibitem{2018_ControlBasedLinear_andaluz}
V.~H. Andaluz, E.~R. S{\'a}sig, W.~D. Chicaiza, and P.~M. Velasco, ``Control
  {{Based}} on {{Linear Algebra}} for {{Mobile Manipulators}},'' in
  \emph{Computational {{Kinematics}}}, S.~Zeghloul, L.~Romdhane, and M.~A.
  Laribi, Eds.\hskip 1em plus 0.5em minus 0.4em\relax {Cham}: {Springer
  International Publishing}, 2018, vol.~50, pp. 79--86.

\bibitem{2019_TrajectoryControlMobile_ram}
R.~Ram, P.~Pathak, and S.~Junco, ``Trajectory control of a mobile manipulator
  in the presence of base disturbance,'' \emph{SIMULATION}, vol.~95, no.~6, pp.
  529--543, 2019.

\bibitem{2022_MotionPlanningMobile_sandakalum}
T.~Sandakalum and M.~H. Ang, ``Motion {{Planning}} for {{Mobile
  Manipulators}}\textemdash{{A Systematic Review}},'' \emph{Machines}, vol.~10,
  no.~2, p.~97, Jan. 2022.

\bibitem{2022_BrunnerElektronikAG_}
``Brunner {{Elektronik AG}},''
  https://www.brunner-innovation.swiss/product/brunner-jet/, 2022.

\bibitem{2017_FullReducedorderModel_ruderman}
M.~Ruderman, ``Full- and reduced-order model of hydraulic cylinder for motion
  control,'' in \emph{{{IECON}} 2017 - 43rd {{Annual Conference}} of the {{IEEE
  Industrial Electronics Society}}}.\hskip 1em plus 0.5em minus 0.4em\relax
  {Beijing}: {IEEE}, Oct. 2017, pp. 7275--7280.

\bibitem{2010_OptimalTrajectoryGeneration_werling}
M.~Werling, J.~Ziegler, S.~Kammel, and S.~Thrun, ``Optimal trajectory
  generation for dynamic street scenarios in a {{Frenet Frame}},'' in
  \emph{2010 {{IEEE International Conference}} on {{Robotics}} and
  {{Automation}}}.\hskip 1em plus 0.5em minus 0.4em\relax {Anchorage, AK}:
  {IEEE}, May 2010, pp. 987--993.

\bibitem{2011_FullyAutonomousDriving_levinson}
J.~Levinson, J.~Askeland, J.~Becker, J.~Dolson, D.~Held, S.~Kammel, J.~Z.
  Kolter, D.~Langer, O.~Pink, V.~Pratt, M.~Sokolsky, G.~Stanek, D.~Stavens,
  A.~Teichman, M.~Werling, and S.~Thrun, ``Towards fully autonomous driving:
  {{Systems}} and algorithms,'' in \emph{2011 {{IEEE Intelligent Vehicles
  Symposium}} ({{IV}})}.\hskip 1em plus 0.5em minus 0.4em\relax {Baden-Baden,
  Germany}: {IEEE}, Jun. 2011, pp. 163--168.

\bibitem{2016_ReviewMotionPlanning_gonzalez}
D.~Gonzalez, J.~Perez, V.~Milanes, and F.~Nashashibi, ``A {{Review}} of
  {{Motion Planning Techniques}} for {{Automated Vehicles}},'' \emph{IEEE
  Transactions on Intelligent Transportation Systems}, vol.~17, no.~4, pp.
  1135--1145, Apr. 2016.

\bibitem{2014_DataAnalysis_brandt}
S.~Brandt, \emph{Data {{Analysis}}}.\hskip 1em plus 0.5em minus 0.4em\relax
  {Cham}: {Springer International Publishing}, 2014.

\bibitem{1993_UsingSignificanceTests_rogers}
J.~L. Rogers, K.~I. Howard, and J.~T. Vessey, ``Using significance tests to
  evaluate equivalence between two experimental groups.'' \emph{Psychological
  Bulletin}, vol. 113, no.~3, pp. 553--565, 1993.

\bibitem{2005_TtestStatisticalEquivalence_limentani}
G.~B. Limentani, M.~C. Ringo, F.~Ye, M.~L. Bergquist, and E.~O. McSorley,
  ``Beyond the t-test: Statistical equivalence testing,'' 2005.

\bibitem{2018_EquivalenceTestingPsychological_lakens}
D.~Lakens, A.~M. Scheel, and P.~M. Isager, ``Equivalence {{Testing}} for
  {{Psychological Research}}: {{A Tutorial}},'' \emph{Advances in Methods and
  Practices in Psychological Science}, vol.~1, no.~2, pp. 259--269, 2018.

\end{thebibliography}
\begin{IEEEbiography}[{\includegraphics[width=1in,height=1.25in,clip,keepaspectratio]{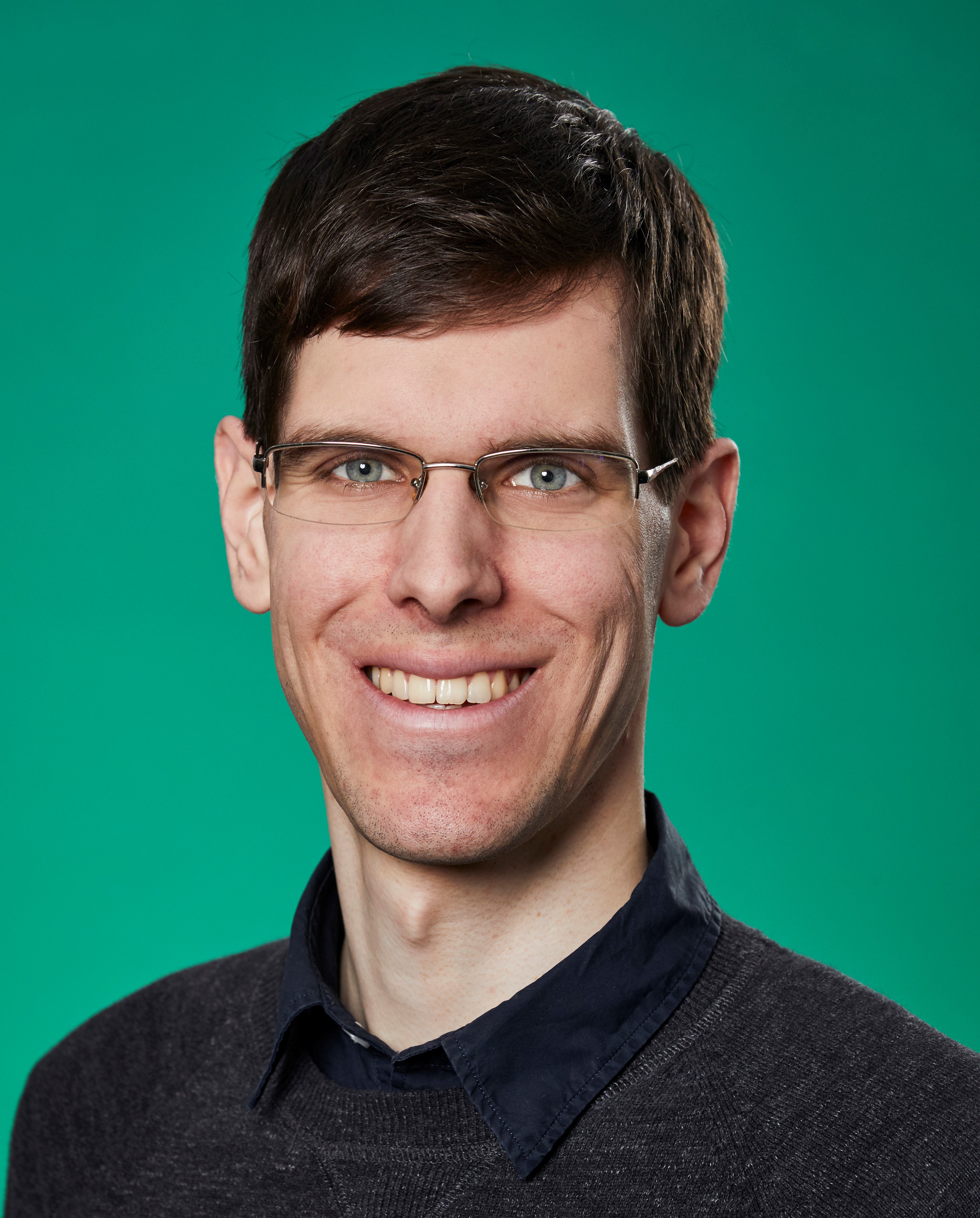}}]{Balint Varga}
received the B.Sc. in mechatronics from the Technical University Budapest, Hungary and M.Sc. degree in mechanical engineering from Karlsruhe Institute of Technology (KIT), Karlsruhe, Germany, in 2016 and 2017, respectively. 2017 -- 2020 research scientist at the FZI Research Center for Information Technology. Since 2020, research assistance at KIT. His research interests include the modelling human-machine interaction and design shared control concepts in the application of vehicle-manipulators.
\end{IEEEbiography}
\vspace*{-2mm}
\begin{IEEEbiography}[{\includegraphics[width=1in,height=1.25in,clip,keepaspectratio]{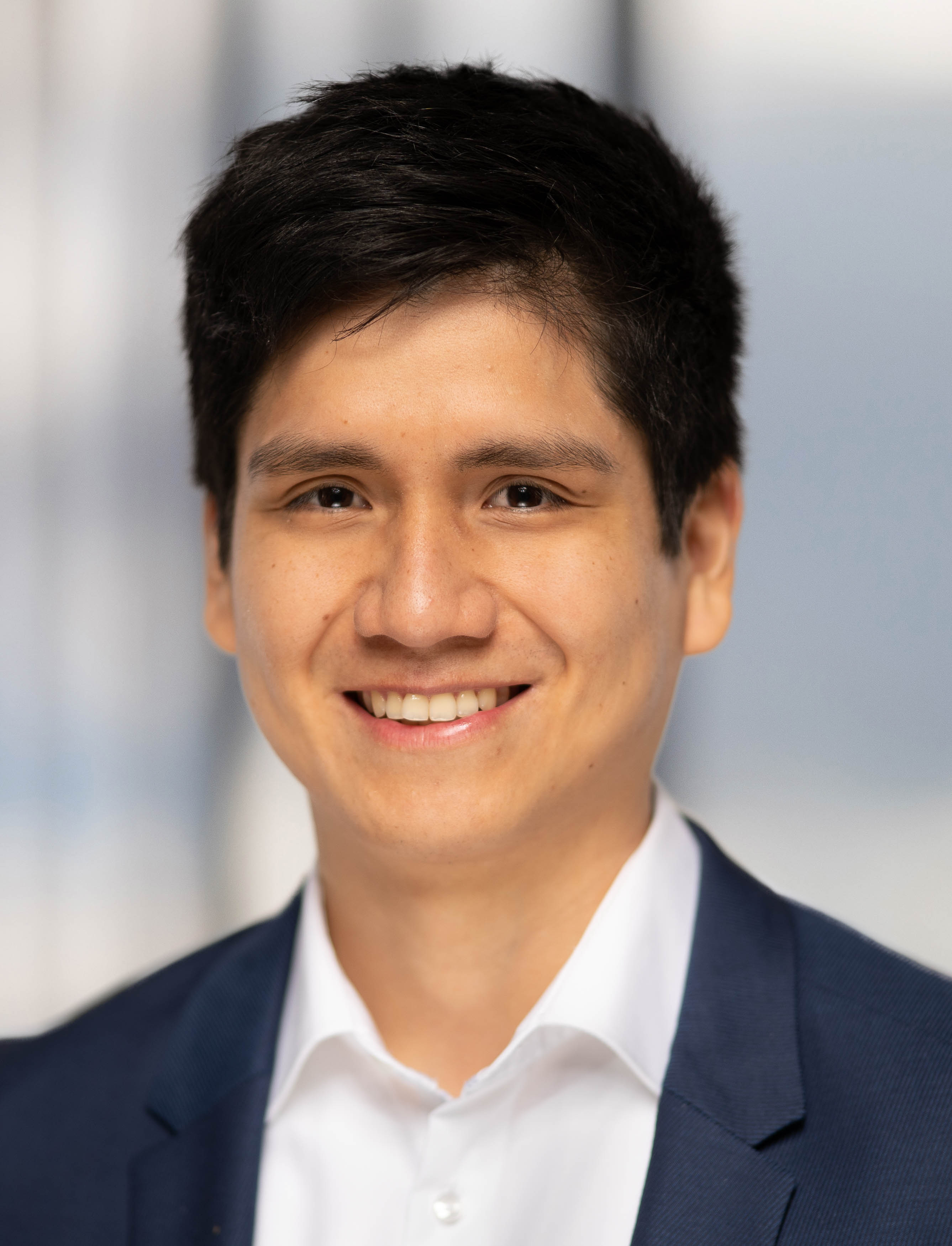}}]{Jairo Inga} received the B.Sc. (2011), M.Sc. (2014), and Dr.-Ing. (2020) in electrical engineering and information technology from the Karlsruhe Institute of Technology (KIT), Germany. He was a research assistant from 2015 to 2020 at the Institute of Control Systems at KIT and currently leads the Cooperative Systems research group at the same institute. His research interests include optimal control and dynamic game theory for modeling and identification of human behavior and the design of cooperative controllers in human-machine systems.
\end{IEEEbiography}
\vspace*{-2mm}
\begin{IEEEbiography}[{\includegraphics[width=1in,height=1.25in,clip,keepaspectratio]{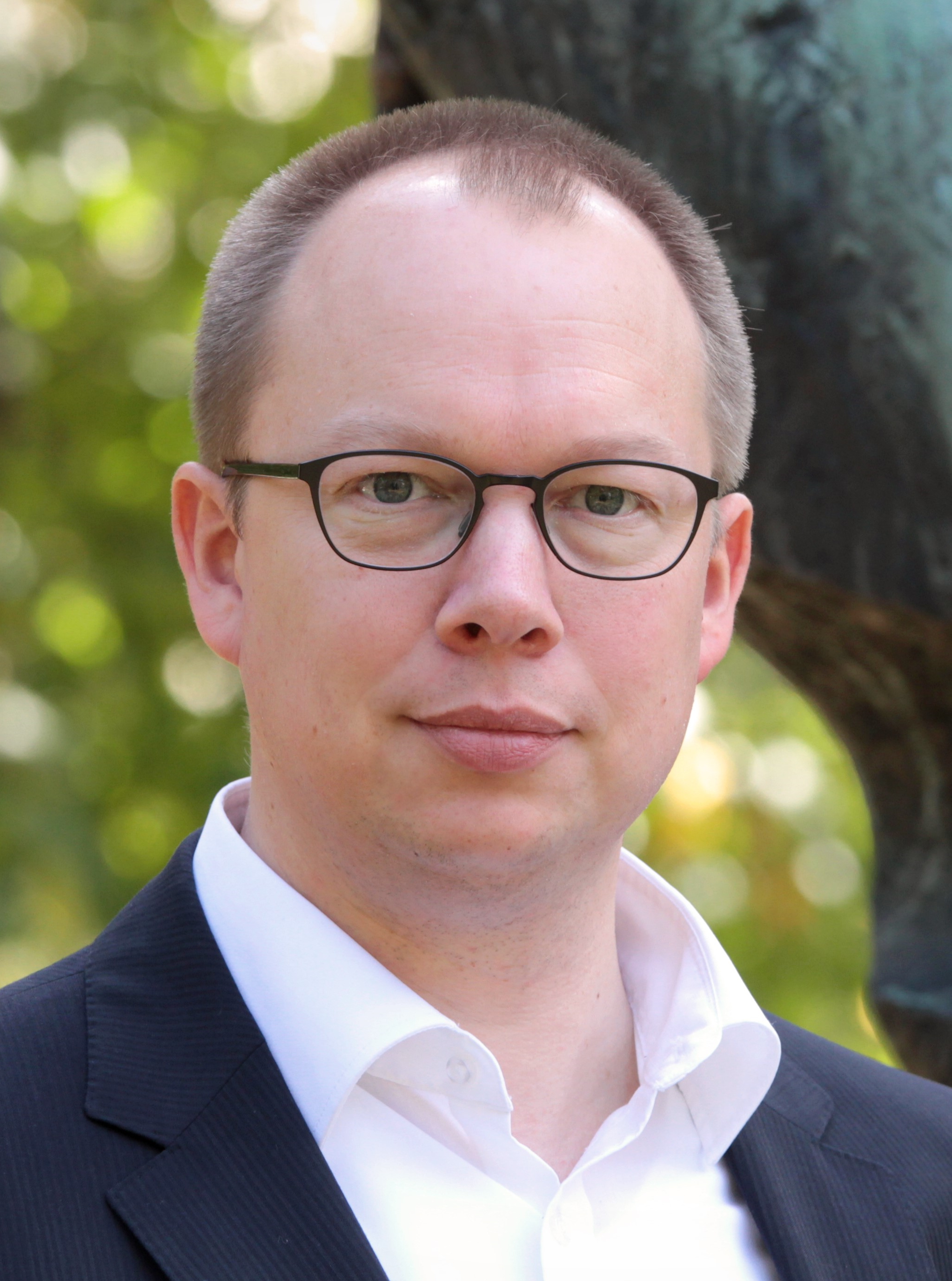}}]{Prof. Dr.-Ing. Sören Hohmann} studied electrical engineering at the Technische Universität Braunschweig, University of Karlsruhe and école nationale supérieure d’électricité et de mécanique Nancy. He received the diploma degree (1997) and PhD degree (2002) from the University of Karlsruhe. Afterwards, until 2010 he worked in the industry for BMW, Munich, where his last position was head of the predevelopment and series development of active safety systems. Today he is the head of the Institute of Control Systems at the Karlsruhe Institute of Technology, Germany as well as a director’s board member of the research center for information technology (FZI), Karlsruhe. His research interests are cooperative control, alternative energies and system guarantees by design.  
\end{IEEEbiography}
\end{document}